\documentclass[journal]{IEEEtran}
\usepackage{caption}
\ifCLASSINFOpdf
\else
   \usepackage[dvips]{graphicx}
\fi
\usepackage{booktabs}
\usepackage{url}
\usepackage{float}
\usepackage{amsmath}
\usepackage{amssymb}
\usepackage{graphicx}
\usepackage{color}
\usepackage{url}
\usepackage[nospace,compress]{cite}
\usepackage{amsbsy}
\usepackage{epsfig}
\usepackage{subcaption}

\makeatletter

\def\bx{{\mathbf x}}

\def\b0{{\mathbf 0}}

\newcommand{\beq}{\begin{equation}}
\newcommand{\eeq}{\end{equation}}

\def\bA{{\mathbf A}}

\def\ba{\mbox{\boldmath $a$}}
\def\bb{\mbox{\boldmath $b$}}

\def\bd{\mbox{\boldmath $d$}}

\def\bn{\mbox{\boldmath $n$}}
\def\bh{\mbox{\boldmath $h$}}
\def\by{\mbox{\boldmath $y$}}

\def\bx{\mbox{\boldmath $x$}}

\def\bs{\mbox{\boldmath $s$}}

\def\bz{\mbox{\boldmath $z$}}
\def\by{\mbox{\boldmath $y$}}

\def\bA{\mbox{\boldmath $A$}}

\def\mb{\mbox{$\mathbf{b}$}}
\def\mA{\mbox{$\mathbf{A}$}}

\def\mB{\mbox{$\mathbf{B}$}}

\def\mC{\mbox{$\mathbf{C}$}}
\def\mD{\mbox{$\mathbf{D}$}}

\def\mD{\mbox{$\mathbf{D}$}}

\def\mH{\mbox{$\mathbf{H}$}}
\def\mI{\mbox{$\mathbf{I}$}}
\def\mL{\mbox{$\mathbf{L}$}}

\def\mP{\mbox{$\mathbf{P}$}}

\def\mU{\mbox{$\mathbf{U}$}}
\def\mV{\mbox{$\mathbf{V}$}}

\def\mW{\mbox{$\mathbf{W}$}}

\newcommand{\ds}{\displaystyle}

\newtheorem{theorem}{\textbf{Theorem}}
\newtheorem{proposition}{Proposition}
\newtheorem{definition}{Definition}

\newenvironment{proof}[1][Proof]{\noindent \textbf{#1.} }{\qedsymbol}
\newcommand{\qedsymbol}{\hspace{\fill}\rule{1.5ex}{1.5ex}}
\floatstyle{ruled}
\newfloat{algorithm}{tbp}{loa}
\providecommand{\algorithmname}{Algorithm}
\floatname{algorithm}{\protect\algorithmname}

\usepackage{amssymb}
\begin{document}

\title{Topological Signal Processing  over\\ Generalized Cell Complexes}

\author{Stefania Sardellitti, \IEEEmembership{Member, IEEE}, Sergio Barbarossa \IEEEmembership{Fellow, IEEE}
\thanks{This work was partially supported  by the H2020 EU/Taiwan Project 5G CONNI, Nr. AMD-861459-3,    by the MIUR under the PRIN Liquid-Edge contract and by the  Huawei Technology France SASU, under agreement N. TC20220919044. The authors are with the Department of Information
Engineering, Electronics, and Telecommunications,
Sapienza University of Rome, Via Eudossiana 18, 00184,
Rome, Italy. E-mails: \{stefania.sardellitti; sergio.barbarossa\}@uniroma1.it. Some preliminary results of this work were presented at the 2021 IEEE  Asilomar Conference \cite{Sar_Bar_Tes}. }}

\maketitle

\begin{abstract}
Topological Signal Processing (TSP) over simplicial complexes is a framework that has been recently proposed, as a generalization of graph signal processing (GSP), to extend GSP to analyzing signals defined over sets of any order (i.e., not only vertices of a graph) and to capture {\it multiway} relations of any order among the data. \textcolor{black}{However, simplicial complexes are required to satisfy the so-called {\it inclusion property}, according to which, if a set belongs to the complex, then all its subsets must also belong to the complex. In some applications, this is a severe limitation. To overcome this limit, in this paper we extend TSP to deal with signals defined over cell complexes and we also generalize the concept of cell complexes to include hollow cells.} 
\textcolor{black}{We show that, even if the algebraic formulation does not change significantly, the extension to the generalized cell complexes considerably broadens the number of applications.  Most important, the new representation provides a much better trade-off between the complexity of the representation and its accuracy.}
In addition, we propose a method to infer the structure of the cell complex from data and we propose distributed filtering strategies, including a method to retrieve the sparsest representation of the harmonic component.
We quantify the advantages of using cell complexes instead of simplicial complexes, in terms of the complexity/accuracy trade-off, \textcolor{black}{for different applications such image segmentation  and  recovering of real flows measured on data traffic and  transportation networks}.
 \end{abstract}


\IEEEpeerreviewmaketitle

\section{Introduction}

\IEEEPARstart{I}{n} the last years  a fast-growing interest emerged in signal processing and machine learning communities to develop models and tools for analyzing data defined over  topological spaces  \cite{carlsson2009topology}, 
i.e. over non-metric domains composed of a set of points along with a set of neighborhood relations.
A relevant example is Graph Signal Processing (GSP) \cite{shuman2013}, 
a recent research field arising at the intersection between signal processing and graph theory, which transposes classical signal processing operations, such as, for example,  filtering, sampling, estimation, to signals defined over the vertices of a graph.
Graphs are a simple example of a topological space able to capture {\it pairwise} (or dyadic) relations between data associated with their vertices by encoding these relations through the presence of edges. However, in many complex systems such as gene regulatory networks, biological,  social and brain networks, the richness of the interactions among the constitutive elements cannot be reduced to simple dyadic relationships 
\cite{lambiotte2019networks}.
To overcome the limitations of graph-based representations, it is necessary to go beyond graphs by incorporating {\it  multiway} (or polyadic) relational descriptors.
This generalization is possible resorting to {\it hypergraphs}, i.e. very general structures composed of a set $\mathcal{V}$ of  elements along with an ensemble $\mathcal{S}$ of subsets of  $\mathcal{V}$ of any order \cite{berge1989hypergraphs}. 
However, the very generality of hypergraphs has an inherent complexity that can limit their applicability. To derive numerically efficient techniques able to handle signals defined over higher-order structures, it is useful to endow the domain with a hierarchical structure that enables processing at different levels. Examples of hypergraphs possessing a hierarchical structure, besides graphs, are {\it simplicial complexes} \cite{munkres2018elements} and {\it cell complexes} \cite{hatcher2005algebraic}, 
\cite{grady2010}. Very recently, combinatorial complexes have also been proposed as a very general hierarchical structure \cite{hajijtopological}. Simplicial complexes handle relations of any order, but are prone to respect the so called {\it inclusion property}, stating that if a set belongs to the space, all its subsets must belong to the space as well. In many applications, this constraint represents an unnecessary and unjustified limitation. Conversely, combinatorial complexes do not need to respect the inclusion property, but they lack of an algebraic structure that is useful in deriving signal processing tools. For all these reasons, in this paper we resort to cell complexes, which are quite general, lead to a natural hierarchical structure and can be formally described through an algebraic framework that provides the formal tools to {\it identify global properties, or invariants, of the signal domain by working with local operators}. 

When analyzing data defined over a topological space, two main approaches are possible. The first class of methods is {\it Topological Data Analysis (TDA)} \cite{carlsson2009topology}, where the data are mapped into a topological space and the information is extracted in terms of properties of the space. Persistent homologies represent one of the fundamental tools in this framework \cite{edelsbrunner2008persistent}. 
Methods based on simplicial complexes have been already applied in many fields, such as statistical ranking \cite{jiang2011}, control systems \cite{muhammad2006control}, tumor progression analysis \cite{roman2015simplicial}, 
and brain networks \cite{andjelkovic2020topology}. The two books \cite{robinson2016topological}, \cite{krim2015geometric} focus on topological and geometric methods to analyze signals and images. The second class of methods is {\it Topological Signal Processing (TSP)}, where the observation is composed of a set of signals (data, features) defined over the subsets of a topological space and the information is extracted by processing these signals using {\it tools tailored to the properties of the space} \cite{barb_2020}, \cite{barb_Mag_2020}, \cite{SCHAUB2021}. Our goal in this paper is to extend the  TSP framework developed in \cite{barb_2020} to signals defined over a broad class of cell complexes.
In TSP, the signals are real values (vectors) associated to the elements of the topological space. For example, there are applications to edge signals, i.e. signals associated to the edge of a network, in gene regulatory networks \cite{lee2018topological}, 
neuronal signals \cite{billings2021simplicial} flowing between different areas of the brain, traffic flows in communication and transportation networks \cite{leung1994traffic} and cross-disciplinary citations \cite{Rosvall1118}.
The analysis of edge signals has also been addressed in previous works as \cite{Schaub2020}, \cite{Segarra_019},\cite{Yang2021}.
In \cite{Schaub2020} the authors introduce a class of filters based on the edge-Laplacian, while in
\cite{Segarra_019}  a semi-supervised
learning method for divergence-free and curl-free edge flows was proposed. In \cite{Yang2021} FIR filters for signals defined over simplices are considered.

A particularly relevant application of TSP is machine learning, and more specifically deep neural networks (DNNs). A message passing  algorithm operating over simplicial complexes was proposed in \cite{bodnar2021weisfeiler}. Examples of DNNs operating over 
signals defined over simplicial complexes 
have been considered in \cite{ebli2020simplicial},\textcolor{black}{
\cite{roddenberry2021principled},\cite{SAN}}. 
An introduction to the processing of signals over cell complexes has been  presented in \cite{Schaub2021CC}, where it was shown that convolutional filters can serve as building blocks in neural networks defined over cell complexes. \textcolor{black}{Recently, in \cite{hajij2022higher}, \cite{Giusti_att} novel attention-based neural networks have been defined over cell complexes. At a more fundamental level, one of the key unresolved issues in DNN is explainability. Some very interesting recent attempts have been proposed in \cite{balestriero2019geometry}, where the authors show that the layers of a large class of DNN can be interpreted as max-affine spline operators (MASOs) that partition the input space and then apply affine mappings to their input to produce the output. This  partition of the input space  can be interpreted as a cell complex. In such an application, the topological approach can be useful to identify invariants of the space that can help in extracting relevant information from the input data.}


\textcolor{black}{Typically, when embedding a cell complex into a real domain, the cells are assumed to be homeomorphic\footnote{A homeomorphism is a mapping  between spaces that preserves all their topological properties. Two topological spaces are homeomorphic if there is a continuous and invertible mapping between them.} to Euclidean balls. 
In this paper, besides extending TSP to conventional cell complexes, we will also introduce a new a generalized class of complexes composed of disjoint polygons that may contain holes.
A motivating example for introducing this new class of complexes is reported in Fig. \ref{fig:Image_monna}. 
\begin{figure}[t!]
\centering
\includegraphics[height=5.8cm,width=5.6cm]{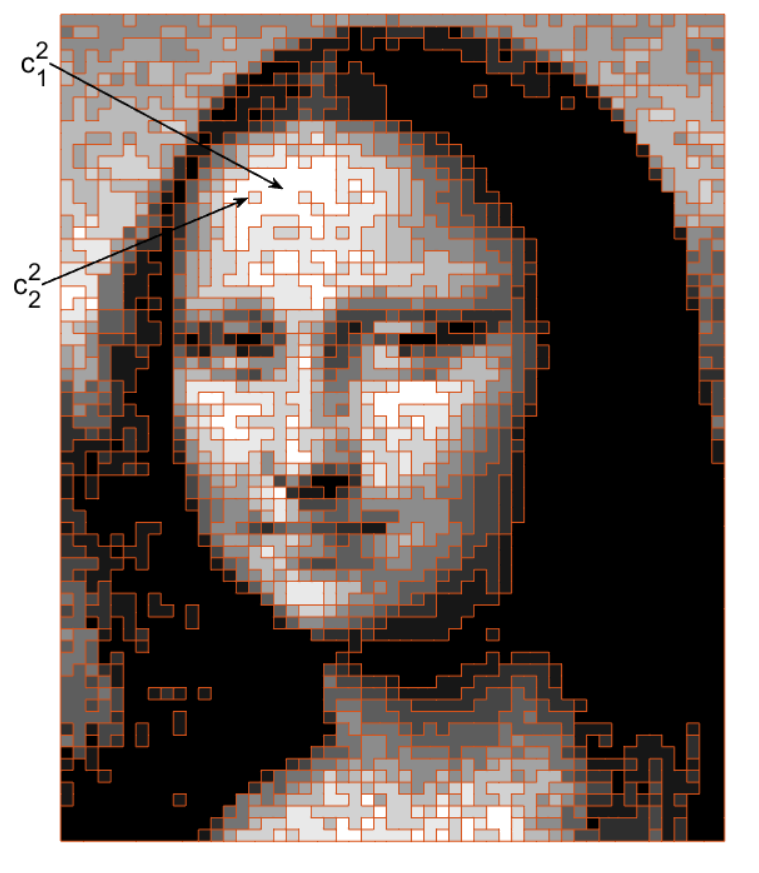}
\caption{Cell complex induced by an image-dependent space partition.}
\label{fig:Image_monna}
\end{figure}
This image has been obtained as follows. We start from a gray-scale image, defined over a rectangular grid within a rectangle ${\cal S}$. Each pixel is located in a position $(x_i, y_i)$ and is the center of a square within which the image assumes a constant value $f(x_i, y_i)$. The ensemble of all the squares represents a partition of the overall domain  ${\cal S}$. In this way, we can think of the image as a 2D (piecewise) function $f(x, y)$, defined over the continuous domain ${\cal S}$. Then, we quantize the image by selecting a discrete number of $H+1$ levels $f_0, f_1, \ldots, f_H$, and we identify the following domains:
\beq \label{eq:domain_i}
{\cal D}_i\equiv \{(x, y): f_i< f(x, y) \le f_{i+1}\}, \, i=0, \ldots, H-1.\eeq 
It is easy to check that the ensemble of all these domains represents a partition of ${\cal S}$. By construction, each domain ${\cal D}_i$ is composed of polygons, possibly disconnected and having holes in their interior. These polygons are examples of generalized cells that, differently from the conventional case, are not homeomorphic to a ball in 2D domain. In computational topology, some algorithms have been developed to find a 2D/3D space partition using a collection of geometric objects that can possibly contain holes, see e.g. \cite{paoluzzi2020topological}, \cite{huhnt2022modeling}, but no TSP tools have been introduced to study signals defined over these domains.} 
This paper is an extension of our preliminary work \cite{Sar_Bar_Tes}, and it contains the following contributions:
\begin{itemize}
  \item[a)] \textcolor{black}{we introduce a new topological structure that generalizes cell complexes to include polygons with holes; the structure is especially suitable in image processing;}
    \item[b)] we show how the proposed complex yields a better rate/distortion trade-off with respect to graphs or simplicial complexes;
    \item[c)] we propose a new method to extract a compact representations of the invariants characterizing the signal domain;
    \item[d)] we propose a method to infer the cell complex structure from data;
    \item[e)] we propose a distributed filtering of signals defined over cell complexes that extends the approaches proposed in \cite{tahbaz2010distributed}, \cite{Yang2021}.
    
\end{itemize}
The paper is organized as follows. In Section II we recall the basic topological tools, distinguishing between  abstract and geometric cell complexes; we also propose a generalization of cell complexes to the case where the (2D) cells may contain holes. In Section III we introduce their algebraic representation, while in Section IV we show how to represent signals defined over cell complexes. In Section V we propose a method to infer the cell complex structure from data and we show how the cell complex model yields a better sparsity/accuracy trade-off with respect to graph- or simplicial-based methods. In Section VI we show how to extract compact representations of the invariants (homologies) associated to the cell structure. In Section VII, we introduce the  design of FIR filters for the solenoidal and irrotational signal components. Finally, in Section VIII we draw some conclusions.

\section{Introduction to algebraic topological tools}
In this section, we recall the basic tools to  represent cell complexes, distinguishing between abstract and geometric cell complexes.
\vspace{-0.2cm}
\subsection{Abstract cell complexes}
\noindent
{As illustrated in  \cite{klette2000cell}, the notion of cell complexes  historically emerged with the rise of topology and is one of the most dynamic notions in the mathematical literature. In this work we focus on complexes composed of a finite number of cells.
An abstract cell complex (ACC) is a finite partially ordered set, or {\it poset}, 
equipped with a dimension function (graded poset \cite{schroder2003ordered}), whose fundamental elements are called cells. A more formal definition of an ACC is the following \cite{klette2000cell}
:}
\begin{definition}
\textit{An abstract cell complex $\mathcal{C}=\{\mathcal{S},  \prec_b, \text{dim}\}$ is a set $\mathcal{S}$ of abstract elements (cells)
provided with a binary relation $\prec_b$, called the bounding (or incidence) relation, and with a dimension function, denoted by $\text{dim}(x)$, that assigns to  each $x \in \mathcal{S}$ a non-negative integer $[x]$, satisfying the following two axioms:
\begin{enumerate}
    \item[A1)] if $x \prec_b y$ and $ y \prec_b z$, then $ x \prec_b z$  follows (transitivity);
    \item[A2)] if $ x \prec_b y$, then $\text{dim}(x)<\text{dim}(y)$ (monotonicity).
\end{enumerate}}
\end{definition}
If $\text{dim}(x)=n$, then $n$ is the dimension (or order) of $x$ and $x$ is  called an $n$-cell; $0$-cells are named vertices.
We denote a cell $x$ of order $n$  with $x^n$.
In words, if $x, y \in \mathcal{S}$ and $x \prec_b y$, we say that $x$ bounds $y$. 
If $x \prec_b y$ and $\text{dim}(x)=\text{dim}(y)-1$ then $x$ is a face of $y$ and $y$ is a co-face of $x$. 
We assume that every $1$-cell of $\mathcal{C}$ is incident with two $0$-cells of $\mathcal{C}$ and distinct $1$-cells are not incident to the same pair of $0$-cells.
Given an $n$-dimensional cell $x^n$, we define its boundary $\partial x^n$ as the set of all cells of dimension less than $n$ that bound $x^n$. An ACC is a $k$-dimensional ACC, or simply a $k$-complex, if the dimensions of all its cells are less than or equal to $k$.
We denote by $\bar{x}^n=x^n \cup \partial x^n$ the closed cell including its boundary.
Two cells $x,y$ are incident if $x \prec_b y$ or $y \prec_b x$. \\
To associate a {\it topology} with an ACC, we need to define neighboring relations among the cells of the complex \cite{Alexandroff}, \cite{barmak2011algebraic}. This can be done by introducing the concept of {\it open} and {\it closed} subsets  \cite{klette2000cell}. 
A subset $U$ of $\mathcal{C}=\{\mathcal{S},\prec_b,\text{dim}\}$ is called {\it open} (or upper set) of $\mathcal{C}$ iff, for every element $x$ of $U$, all elements $y$ of $\mathcal{C}$  for which $ x \prec_b y$ are also in $U$. A subset $U$ is  {\it closed} (or down set) iff, for every element $x$ of $U$, all elements $y$ of $\mathcal{C}$ for which $ y \prec_b x$ are also in $U$.\\
An ACC equipped with the above neighboring relations is a topological space \cite{barmak2011algebraic}.
In signal processing, it is also useful to associate an algebraic structure with the topological space. To do that, we need to introduce
an algebraic operator, called the {\it boundary operator}, which maps each cell to its boundary. 
This mapping is made possible if we assume that  all cells belonging to the complex are \textit{closed}: A cell $x$ of order $n$ is closed if   all its bounding faces belong to the complex. The boundary $\partial x$ is the sequence of all cells of order $n-1$ that bound $x$.
The boundary operator identifies each cell uniquely and makes possible the association of an algebraic representation with an abstract cell complex.\\   
As an example, let us consider a simple ACC composed by $4$  cells of order $0$, $4$  cells of order $1$ and $1$ cell of order $2$. Denoting  the $i$-th cell of order $k$ as $c_i^{k}$,  the set $\mathcal{S}$ of abstract elements is defined by 
$\mathcal{S}=\{ \{c_i^{0}\}_{i=1}^{4}, \{c_i^{1}\}_{i=1}^{4}, c_1^{2}\}$. The neighboring relations among the elements of  the set $\mathcal{S}$ are given by the four open sets $U_1=\{ c_1^{0},c_1^{1},c_4^{1},c_1^{2}\}$, $U_2=\{ c_2^{0},c_4^{1},c_3^{1},c_1^{2}\}$, $U_3=\{ c_3^{0},c_1^{1},c_2^{1},c_1^{2}\}$, $U_4=\{ c_4^{0},c_2^{1},c_3^{1},c_1^{2}\}$.  These open sets  fully describe the cell complex $\mathcal{C}$. Note that  the boundary $\partial c_1^{2}$ of the second order cell $c_1^{2}$ is defined by the closed sequence of 
incident edges  $[c_1^{1},c_2^{1},c_3^{1},c_4^{1}]$, where all edges belong to the complex. Since the cell $c_1^{2}$ is closed, it can be uniquely identified through its bounding edges.   
Note that the notion of abstract cell complex differs and it is much more general than that of a simplicial complex because it is not required to respect the inclusion property. 
\subsection{Geometric cell complexes}
Differently from ACCs, a \textit{geometric cell complex} (GCC) is a collection of finite-dimensional cells embedded in a Euclidean space. A geometric cell complex is a topological space built by sequentially attaching cells of increasing dimension along their boundaries.
In the following, we first introduce  regular geometric cell complexes whose cells are homeomorphic to Euclidean balls. Then, we will define more  general cell complexes whose cells are not constrained to be homeomorphic to Euclidean balls but may contain holes.\\
\textbf{Regular geometric cell complexes (RGCCs).}
The formal definition  describing how the cells may be glued together in an RGCC, is the following \cite{hansen2019toward}:
\begin{definition}
A \textit{regular cell complex} is a topological space $X$ with a partition  $\{X_{\sigma}\}_{\sigma \in P_{X}}$ of subspaces $X_{\sigma}$ of $X$, called cells, \textcolor{black}{where $P_{X}$ is a partially ordered set},  satisfying the following conditions:
\begin{enumerate}
  \item[1)]\textcolor{black}{every {\it open} $k$-order cell $X_\sigma$ is homeomorphic to the interior of a $k$-dimensional Euclidean ball;}
  \item[2)] denoting by $\partial X_\sigma$ the boundary of $X_\sigma$, every {\it closed} $k$-order cell can be written as $\bar{X}_\sigma=X_\sigma \cup \partial X_\sigma$, and it is homeomorphic to a closed ball in a $k$-dimensional Euclidean space; $\partial X_\sigma$ is then homeomorphic to the closure of the ball associated to $X_\sigma$;
    \item[3)] For all $\alpha$, $\sigma$,   $\bar{X}_{\alpha} \cap X_{\sigma} \neq \emptyset$ iff $X_{\sigma}\subseteq \bar{X}_{\alpha}$.
\end{enumerate}
\end{definition}
\textcolor{black}{Since any two open cells are disjoint by definition, condition 3) implies that two closed cells can only intersect over their boundaries.}
Intuitively speaking, a regular cell complex allows us to represent a Euclidean space as a poset, i.e. a hierarchical structure that will turn out to be very useful in the analysis of signals defined over the structure.
Every cell $c_i$ has associated a number, called dimension. The $n$-skeleton ${X}^n$ of order $n$ of ${X}$ is defined as ${X}^n=\underset{ dim(c_i)\leq n}{\bigcup} c_i$. The boundary of a cell $c$ of dimension $n$ is the set of all cells of dimension $n-1$  that bound $c$. As with ACCs, we assume that every $1$-cell of $\mathcal{C}$ is incident with two $0$-cells of $\mathcal{C}$ and distinct $1$-cells are not incident to the same pair of $0$-cells. \\ As an example, a closed $1$-cell (edge) can be built from two $0$-cells (points $c_1^0, c_2^0$) and an interval $I$ ($1$-dimensional ball) by glueing one endpoint of $I$ to $c_1^0$ and the other to $c_2^0$.
A $k$-cell may be represented by an ordered set of vertices comprising a  $k$-polytope\footnote{ A $k$-polytope is a geometric object with flat sides, i.e. with sides consisting of $(k-1)$-polytopes that may have $(k-2)$-polytopes in common.}.
 The dimension or order of a cell complex is the largest dimension of any of its cells. 
A graph is a particular case of a  cell complex of order $1$, containing only cells of order $0$ (nodes) and $1$ (edges).
Simplicial complexes are   special cases of cell complexes
 for which any $k$-cell is composed of exactly $k+1$-vertices.
In Fig. \ref{fig:ACC_2} we show an example of a $3$-dimensional GCC  defined by the set of elements $\mathcal{S}=\{ \{c_i^0\}_{i=1}^{18},  \{c_i^1\}_{i=1}^{25}, \{c_i^2\}_{i=1}^{8} ,c_1^{3}\}$.
The illustrated structure  is not a simplicial complex, because the presence of the nonagon (cell $c_1^2$) and the closed cube (cell $c_1^3$) does not imply the presence of all their subsets. Notice also that the $2$-order closed cell $\bar{c}_1^2$ contains more  $0$-order cells (vertices)  than the $3$-order closed cell $\bar{c}_1^3$. This is a situation that cannot occur in a simplicial complex.\\
 \begin{figure}[t]
\centering
\includegraphics[width=7.1cm,height=4.6cm]{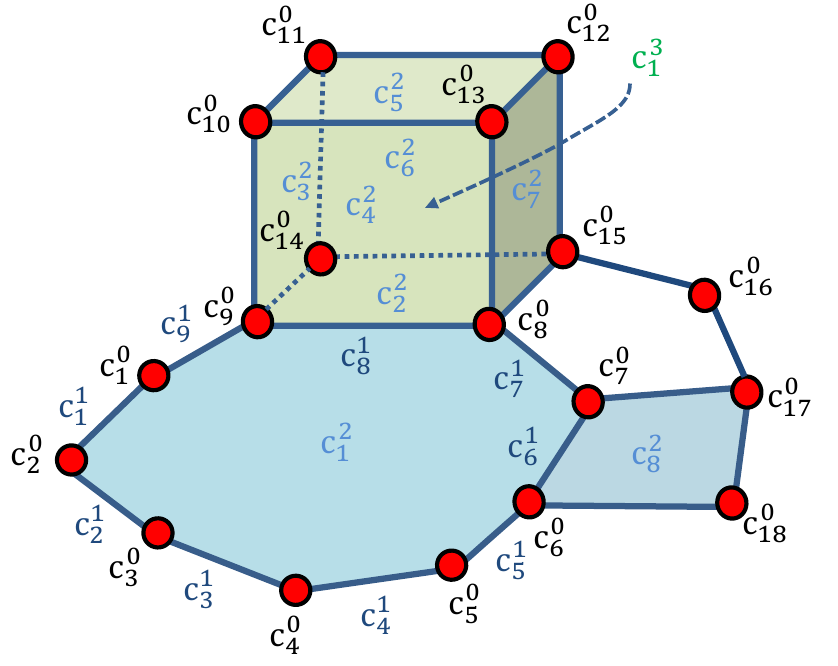}
\caption{A RGCC of order  $3$.}
\label{fig:ACC_2}
\end{figure}
\textcolor{black}{\textbf{Planar hollow cell complexes}.
The previous definition of geometric cell complex does not encompass the complex depicted in Fig. \ref{fig:Image_monna}, where some polygons contain holes. In this section, we generalize the previous definition of complex to be able to handle signals as in Fig.  \ref{fig:Image_monna}. More specifically, we introduce a new class of complexes, that we name  {\it planar hollow cell complex}, whose distinguishing feature is that its cells are not constrained to be homeomorphic to a 2D ball, but can  contain a finite number of holes in their interior.
More specifically, we use the following definition:
\begin{definition} Given a finite region ${\cal D}$ of $\mathbb{R}^{2}$, a Planar Hollow Cell Complex (PHCC) is a partition of  ${\cal D}$ with a finite number of cells  $\{X_{\sigma}\}_{\sigma \in P_{X}}$ satisfying the following conditions:
\begin{enumerate}
         \item[1)] The $1$-skeleton associated with $X$ is a planar graph\footnote{A planar graph is a graph that can be embedded in the plane in a such a way that its edges intersect only at their endpoints.};
 \item[2)] Every $k$-cell $X_{\sigma}$, with $k=0, 1$, is homeomorphic to the open $k$-ball in $\mathbb{R}^k$;
  \item[3)] Every hollow $2$-cell $X_{\sigma}$ is a path-connected region that is homeomorphic to a $2$-ball in $\mathbb{R}^{2}$ that may include a finite number $n_h$ of 2D balls, with $n_h\ge 0$.
\textcolor{black}{\item[4)]
The boundary of a hollow $2$-cell, including $n_h$ holes in its interior, is composed of $n_h+1$
 disconnected  closed paths, each one belonging to a different graph component.
 Every cell is a path-connected space enclosed between the outer and the $n_h$ inner boundaries of the $n_h$ holes; each hole may be filled with other cells or it may be empty.
}
         \item[5)] A closed cell  $\bar{X}_{\alpha}$   is defined as $\bar{X}_{\alpha}=X_{\alpha} \cup \partial X_{\alpha}$. For cells $X_{\alpha}$ of order $k=0,1$ the boundary  $\partial X_{\alpha}$ is homeomorphic to the closure of the ball associated with $X_{\alpha}$. 
         For all $\alpha$, $\sigma$,   $\bar{X}_{\alpha} \cap X_{\sigma} \neq \emptyset$ iff $X_{\sigma}\subseteq \bar{X}_{\alpha}$.   
       \end{enumerate}
\end{definition}
}
\textcolor{black}{As a pictorial example,  Fig. \ref{fig:Cell_part}(a)  shows a PHCC composed of a hollow cell with one hole (the blue cell $c_1^{2}$), filled with  two triangular cells with no holes (the green cells $c_2^{2}$ and $c_3^{2}$).}
\begin{figure}[t]
\centering
\includegraphics[width=5.0cm,height=4.2cm]{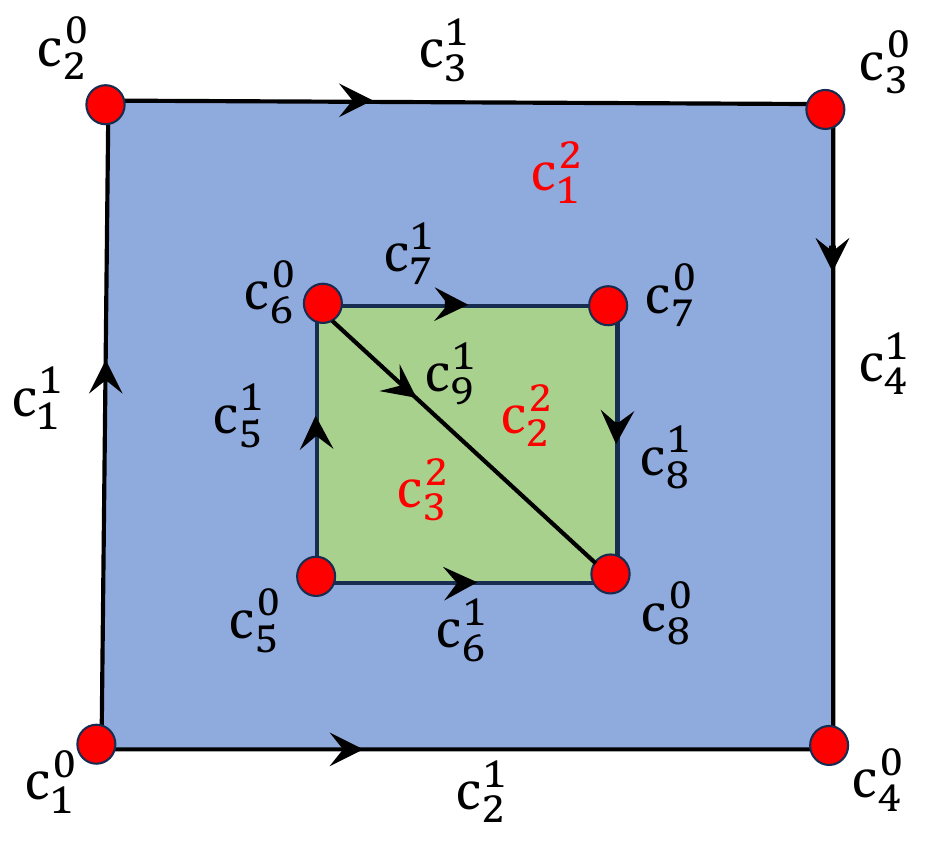}\\
(a)\\
\vspace{0.2cm}
\hspace{0.4cm}\includegraphics[width=5.5cm,height=4.2cm]{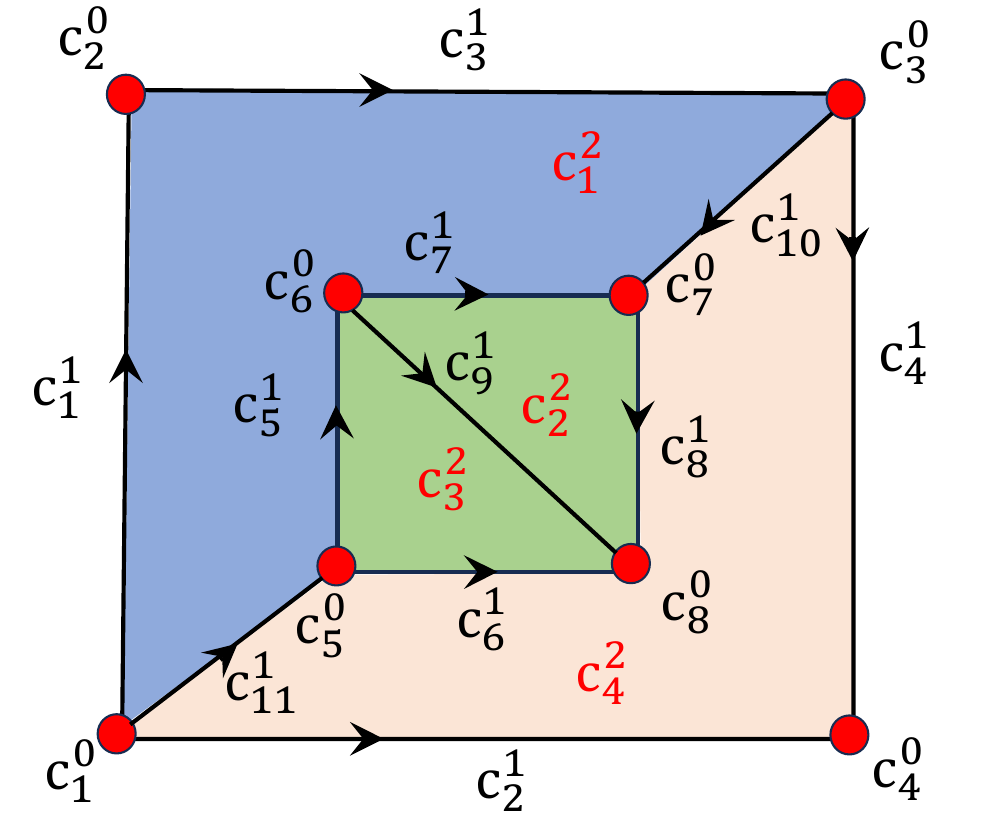}\\
(b)
\caption{(a) A planar hollow cell complex; (b) a planar polygonal cell complex.}
\label{fig:Cell_part}
\end{figure}

\section{Algebraic representation of cell complexes}
The structure of a cell complex is captured by the neighborhood relations among its cells.
As with graphs, it is useful to introduce the orientation of the cells  first. The orientation of a cell can be derived by generalizing the concept of orientation of a simplex.
Every simplex can only  have two orientations, depending on the permutations of its elements. Defining the transposition as the permutation of two elements, two orientations are equivalent if each of them can be recovered from the other through an even number of transpositions \cite{munkres2000topology}.\\
To define the orientation of $k$-cells, we may apply first a simplicial decomposition \cite{grady2010}, which consists in  subdividing the cell into a set of internal $k$-simplices, so that: i) two simplices share exactly one $(k-1)$-simplicial boundary element, which is not the boundary of any other $k$-cell in the complex; and ii) two $k$-simplices induce an opposite orientation on the shared $(k-1)$-boundary. Hence, by orienting a single internal simplex, the orientation propagates on the entire cell.\\ 
\textbf{Abstract Cell complexes and GCCs}. 
An oriented $k$-cell of an ACC or a GCC may be represented through its bounding cells as $c^k=[c^{k-1}_1,\ldots,c^{k-1}_M]$  where two consecutive $(k-1)$-cells, $c^{k-1}_{i}$ and $c^{k-1}_{i+1}$ share a common $(k-2)$-cell boundary.
We use the notation $c^{k-1}_i \sim c^{k}_j$ to indicate that the orientation of $c^{k-1}_i$ is coherent with that of  $c^{k}_j$ and $c^{k-1}_i \nsim c^{k}_j$ to indicate that their orientations are opposite.
Given an orientation, there are two ways in which two cells can be considered to be adjacent: lower and upper adjacent.
Two $k$-order cells are lower adjacent if they share a common face of order $k-1$ and upper adjacent if they are both faces of a cell of order $k+1$.
Given an orientation of the cell complex ${\cal C}$, the structure of a cell complex of order $K$ is fully captured by the set of its incidences matrices $\mB_k$ with $k=1,\ldots,K$, also named boundaries matrices, whose entries  establish which $k$-cells are incident to which $(k-1)$-cells and are defined as follows:
  \beq \label{inc_coeff}
  B_k(i,j)=\left\{\begin{array}{rll}
  0, & \text{if} \; c^{k-1}_i \not\prec_b c^{k}_j \\
  1,& \text{if} \; c^{k-1}_i \prec_b c^{k}_j \;  \text{and} \; c^{k-1}_i \sim c^{k}_j\\
  -1,& \text{if} \; c^{k-1}_i \prec_b c^{k}_j \;  \text{and} \; c^{k-1}_i \nsim c^{k}_j\\
  \end{array}\right. .
  \eeq

Let us consider a cell complex of order two  $\mathcal{C}=\{\mathcal{V},\mathcal{E},\mathcal{P}\}$
where $\mathcal{V}$, $\mathcal{E}$, $\mathcal{P}$ denote the set of  $0$, $1$ and $2$-cells, i.e. vertices, edges and polygons, respectively. We denote their cardinality by $|\mathcal{V}|=N$, $|\mathcal{E}|=E$ and $|\mathcal{P}|=P$.
Then, the two incidence matrices describing the connectivity of the complex are $\mB_1 \in \mathbb{R}^{N\times E}$ and $\mB_2 \in \mathbb{R}^{E\times P}$, where $\mB_2$ can be written as
\beq \label{eq:B2_cell}
\mB_2=[\mB_{T},\mB_{Q},\ldots,\mB_{P_{max}} ]
\eeq
with  $\mB_{T}$, $\mB_{Q}$ and  $\mB_{P_{max}}$ describing the incidences between edges and, respectively, triangles, quadrilateral, up to polygons with  $P_{max}$ sides. To build the incidence matrices
we need to find  the edges  bounding  each  polygon in the complex. Defining a cycle  as a sequence  of distinct and lower adjacent edges, starting from  and ending in the same node, we have to search for all chordless cycles. To do that, we first find all cycles  of increasing length, from $3$ up to a maximum value $P_{max}$, by using a   numerical graph toolbox, for example, the matlab function allcycles. Then, we eliminate every cycle of length $n$  that can be obtained as a linear combination of cycles  of length  smaller than $n$. The remaining cycles identify the boundary edges of each $2$-cell and are then used  to build the columns of the incidence matrices in (\ref{eq:B2_cell}).\\
An important property is that the boundary of a boundary is zero, i.e. it always holds $\mB_k \mB_{k+1}=\mathbf{0}$.
To describe the structure of the $K$-cell complex we can consider the higher order combinatorial Laplacian matrices 
\cite{Horak13} 
given by
\beq
\begin{split}
& \mL_0=\mB_1\mB_1^T,\\
&\mL_k=\mB_k^T\mB_k+\mB_{k+1}\mB_{k+1}^T \; \; \mbox{for} \; k=1,\ldots,K-1\\
&\mL_K=\mB_K^T\mB_K
\end{split}
\eeq
where $\mL_{k,d}=\mB_k^T \mB_k$ and $\mL_{k,u}=\mB_{k+1} \mB_{k+1}^T$ are  the lower and upper Laplacians, expressing the lower and upper adjacency of the $k$-order cells, respectively.
Note that $\mL_0$ is the graph Laplacian.
There are some interesting properties of  the eigenvectors of higher order Laplacians  \cite{Horak13}, 
that it is  useful to recall for our ensuing spectral analysis. Considering w.l.o.g. the first order Laplacian, i.e.
\beq
\label{L1up+L1down}
\mL_1=\mL_{1,d}+\mL_{1,u}=\mB_1^T \mB_1+\mB_2 \mB_2^T,
\eeq
it holds: i) the eigenvectors associated  with  the  nonzero  eigenvalues  of  $\mL_{1,d}=\mathbf{B}^T_1\mathbf{B}_1$  are  orthogonal  to  the eigenvectors associated with the nonzero eigenvalues of $\mL_{1,u}=\mathbf{B}_{2}\mathbf{B}_{2}^T$ and viceversa; ii)
the eigenvectors  associated with the nonzero eigenvalues $\lambda^{1}_i$ of  $\mathbf{L}_1$ are either the eigenvectors of $\mL_{1,d}$ or those of $\mL_{1,u}$; and, finally, iii)
the  nonzero eigenvalues of $\mathbf{L}_1$ are either the eigenvalues of $\mL_{1,d}$ or those of $\mL_{1,u}$.
This spectral structure of $\mL_1$ induces an interesting decomposition of the whole space $\mathbb{R}^{E}$, the so-called Hodge decomposition \cite{Lim}, given by
\beq \label{eq:Hodge_dec}
\mathbb{R}^{E} \triangleq \text{img}(\mB_1^T) \oplus \text{ker}(\mL_1)\oplus \text{img}(\mB_2)
\eeq
where the vectors in $\text{ker}(\mL_1)$ are also in
$\text{ker}(\mB_1)$ and $\text{ker}(\mB_2^T)$.\\

\textcolor{black}
{\textbf{Planar  Hollow Cell Complexes.}
In this section,
we generalize the standard cell complex approach to incorporate cells (polygons) containing one or more holes in their interior. What distinguishes hollow cells from the others is that they have one outer boundary and one or more inner boundaries, one for each hole. These boundaries are, by construction, cycles  that are disconnected from each other. In this section, we show that also working with hollow cells, we can still  extract the fundamental invariant properties of the space, such as the number of connected components or number of holes, from the algebraic representation given by the incidence matrices. 
We start with the illustrative example of the $2$-order complex depicted in Fig. \ref{fig:Cell_part}(a), and then we will provide the full generalization. The complex in Fig.  \ref{fig:Cell_part}(a) has an outer square cell with one inner hole, containing two triangular cells without holes. In this case, the node-edge and edge-polygon incidence matrices are, respectively,  $$\mB_1=\left[\begin{array}{ccccccccc} 
-1 & -1 & 0 & 0 & 0 & 0 &  0& 0 & 0\\  
 1 & 0 & -1 & 0 & 0 & 0 &  0& 0 & 0\\
 0 & 0 & 1 & -1 & 0 & 0 &  0 & 0 & 0\\ 
 0 & 1 & 0 & 1 & 0 & 0 &   0 & 0 & 0\\ 
 0 & 0 & 0 & 0 & -1 & -1 & 0& 0 & 0\\ 
 0 & 0 & 0 & 0 & 1 & 0 &   -1 & 0 & -1\\ 
 0 & 0 & 0 & 0 & 0 & 0 &  1 & -1 & 0\\ 
 0 & 0 & 0 & 0 & 0 & 1 &  0 & 1 & 1
\end{array}\right],$$ $$\mB_2=\left[\begin{array}{rrr} 0 & 0 & 1\\  0 & 0 & -1\\
 0 & 0 & 1\\ 0 & 0 & 1\\ 0 & 1 & -1\\ 0 & -1 & 1\\ 1 & 0 & -1\\ 1 & 0 & -1\\ -1 & 1 & 0
\end{array}\right]
$$ where the first two columns of $\mathbf{B}_{2}$
identify the triangular cells $c_2^{2}$ and $c_2^{3}$ without holes, while the third column identifies the cell $c_1^{2}$  through its outer and inner bounding edges. 
Looking at $\mB_1$, we can see that it has a block structure, corresponding to the two disconnected graphs composing the $1$-skeleton of the complex. The kernel of $\mB_1$ has then a dimension $2$, coinciding with the number of disconnected components. 
Furthermore, it is easy to verify that, as with standard cell complexes,  $\mB_1\mB_2=0$.
This implies, using standard algebra, see e.g., \cite{Lim}, that 
$${\rm im}(\mB_1^T \mB_1+\mB_2 \mB_2^T)={\rm im}(\mB_1^T)+{\rm im}(\mB_2).$$
In this example, the dimension of ${\rm im}(\mB_1^T)$ is $8-2=6$, while the dimension of  ${\rm im}(\mB_2)$ is $3$, since $\mB_2$ is clearly full rank. Hence, the dimension of ${\rm im}(\mL_1)$ is $6+3=9$, which means that $\mL_1$ is full rank. This implies that the number of holes, i.e. zero,  is correctly identified with the dimension of the  kernel of $\mL_1$. 
Furthermore, according to Euler's Theorem for planar graphs with $n_c$ disconnected components \cite{gallier2011discrete}, the number of $2$-order faces in $\mathcal{X}$, excluding the unbounded external face,  is given by $F=E-N+n_c$. Therefore, in our example we get $F=3$. 
Moreover, we can check that, if we remove any face from the space, which means that we remove one column from the incidence matrix $\mB_2$, we create a hole in the space and indeed the dimension of 
$\text{ker}(\mL_{1})$ becomes equal to one, i.e. it coincides with the number of holes. In summary, the first two Betti numbers of the space sketched in Fig. \ref{fig:Cell_part}(a) are then $\beta_0=2$ (number of connected components) and $\beta_1=0$ (number of holes). An interesting property to notice about filling this space with hollow cells is that if we had filled the same space with a structure whose $1$-d skeleton was connected, as depicted in Fig. \ref{fig:Cell_part}(b), we would have found $\beta_0=1$ and $\beta_1=0$.}\\
\textcolor{black}{We prove now that the above results are valid in general, for any planar hollow cell complex. 
\begin{theorem}
    Assume that a $2$-order cell complex  $\mathcal{X}=\{\mathcal{V},\mathcal{E},\mathcal{F}\}$  defines   a partition of a   closed space $\mathcal{D}\subseteq \mathbb{R}^2$ through its cells $c_i^{k}$    of order $k=0,1,2$,     with  $\bigcup_{i,k} c_i^{k}=\mathcal{D}$. The $1$-skeleton associated with the complex is composed by the union of $n_c \geq 1$ disconnected     graphs $\mathcal{G}_i=\{\mathcal{V}_i, \mathcal{E}_i\}$. 
    Then, it holds:
    \begin{itemize}
        \item [a)] The cell complex $\mathcal{X}$ admits an algebraic representation through the node-edge and edge-polygon  incidence matrices $\mB_1$ and $\mB_2$ given, respectively,  in (\ref{eq:B1_bl}) and (\ref{eq:B_2_planar}) 
           and the  following orthogonality condition  holds  
        \beq
        \mB_1 \mB_2=\mathbf{0}.
        \eeq
        \item [b)]  The structure of the  planar complex  is described by the first-order  Laplacian matrix
       $\mL_1=\mB_1^T \mB_1+\mB_2 \mB_2^T,$ which implies that  the Hodge decomposition still holds, so that the whole space $\mathbb{R}^E$ is the direct sum of three orthogonal subspaces, the irrotational, solenoidal and harmonic subspaces, as in  (\ref{eq:Hodge_dec}).
       The dimension of $\text{ker}(\mL_1)$ counts the number of  $1$-dimensional holes in the complex, i.e. the regions of the plane that remain uncovered.   
    \end{itemize}
\end{theorem}
\begin{proof} See Appendix A. \end{proof}}\\
\section{Signal processing over cell complexes}
Some fundamental tools to analyze signals defined over simplicial complexes are provided in \cite{grady2010}, \cite{barb_2020}, \cite{barb_Mag_2020}, \cite{SCHAUB2021}. Here we show that the framework can be easily extended to cell complexes. We focus w.l.o.g. on cell complexes of order $2$. Given a cell complex $\mathcal{C}=\{\mathcal{V},\mathcal{E},\mathcal{P}\}$, the signal on vertices, edges and polygons are defined by the following maps: $\bs^0: {\cal V} \rightarrow \mathbb{R}^N$, $\bs^1: {\cal E} \rightarrow \mathbb{R}^E$, and $\bs^2: {\cal P} \rightarrow \mathbb{R}^P$.
Exploiting the Hodge decomposition in (\ref{eq:Hodge_dec}), we may always express a signal $\bs^1$ as
the sum of three orthogonal components \cite{Lim}, i.e. 
\beq \label{eq:s1_dec}
\bs^1=\mB_1^T \bs^0+ \mB_2 \bs^2+\bs^1_H.
\eeq
In analogy to vector calculus terminology, the first component $\bs^1_{irr}:=\mB_1^T \bs^0$ is called the irrotational component since, using the equality $\mB_1 \mB_2=\mathbf{0}$, it has zero curl, i.e. $\mB_2^T \bs^1_{irr}=\mathbf{0}$. The second term $\bs^1_{sol}:= \mB_2 \bs^2$  is the solenoidal component, since its divergence defined as $\mB_1 \bs^1_{sol}$ is zero. Finally, the component $\bs^1_H$ is the harmonic component of the signal since it belongs to $\text{ker}(\mL_1)$, and it is a signal with zero curl and zero divergence.
\textcolor{black}{Note  from (\ref{eq:s1_dec}) that the signals defined over cells of order $k$ depend on the signals defined over cells of order $k-1$ and $k+1$. This interdependence between signals of different orders is what allows us to encode and capture the topological structure of the data.}\\
 A useful orthogonal basis to represent signals of various order, capturing the connectivity properties of the complex, is given by the eigenvectors of the corresponding higher order Laplacian.
 Then, generalizing graph spectral theory \textcolor{black}{\cite{shuman2013}, we can introduce, as for simplicial complex \cite{barb_2020},} a notion of Cell complex Fourier Transform (CFT) for signals defined over cell complexes.
 Let us consider the eigendecomposition
\beq
\mL_k= \mU_k \boldsymbol{\Lambda}_k \mU_k^T
\eeq
where $\mU_k$ is the eigenvectors matrix and $\boldsymbol{\Lambda}_k$ is a diagonal matrix with entries the eigenvalues $\lambda_i^k$ of $\mL_k$, with $i=1,\ldots, E$. Then, we define the $k$-order CFT as the projection of a $k$-order signal onto the eigenvectors of $\mL_k$, i.e.
\beq \label{eq:GFT}
\hat{\bs}^k \triangleq \mU_k^T \bs^k. \eeq
A signal $\bs^k$ can then be represented in terms of its CFT coefficients as $\bs^k =\mU_k \hat{\bs}^k.$
\textcolor{black}{Note that   the eigenvectors of the upper and lower Laplacians  are the optimal orthogonal vectors minimizing the quadratic forms $\mathbf{x}^{1\, T} \mB_1^T \mB_1 \mathbf{x}^1$ and $\mathbf{x}^{1\, T} \mB_2 \mB_2^T \mathbf{x}^1$, respectively. These quadratic forms represent the squared norm of the divergence and the curl vectors of 
$\mathbf{x}^1$, respectively. 
Then,  the first eigenvectors of the lower and upper Laplacians enable capturing the signals smoothness  in terms of divergence and curl variations.} \\
In \cite{barb_2020}, using the signal decomposition in  (\ref{eq:s1_dec}) we analyzed optimization strategies for the estimation of the signals $\bs^0$, $\bs^2$ and $\bs_H^{1}$ from noisy observations of edge signals defined over simplicial complexes and  we provided conditions for their sampling and reconstruction. Those results can be extended to signals defined over cell complexes because the structure of $\mB_2$ in (\ref{eq:B2_cell}) 
 and  (\ref{eq:B_2_planar}) 
has the same properties, thus generalizing the approach of  \cite{barb_2020}.\\
\textcolor{black}{Let us now investigate how the use of cell complexes yields a substantial performance  improvement in terms of sparsity/accuracy trade-off with respect to simplicial structures.} 
\textcolor{black}{As an illustrative application, we focus on an image coding task by showing how planar cell complexes are capable of partitioning the image domain in cells encoding similarity among the intensity of adjacent pixels. 
We consider, as an example, a gray-scale version of the Mona Lisa painting, shown in Fig. \ref{fig:Image_monna}.  The  image is defined over a rectangular grid by associating  each pixel to a square. 
Then, we first derive a  simplicial representation of the image by   building a 
gradient-based Delaunay
triangulation of its domain  using $6100$ triangles and   $12$  gray levels. Hence, we get a 2D-dimensional function $f(x,y)$ with constant intensity (gray level) over each triangle. 
We denote with $\bs_2(i)$ the signal  associated to  each triangle $i$, i.e.  the image intensity  observed over it. 
Then, 
we build a planar cell 
complex whose cells are identified according to (\ref{eq:domain_i}) where $f_i$, $i=0,\ldots,H$, are the selected gray levels. 
The result, as shown in Fig. \ref{fig:Image_monna}, is a planar hollow cell complex composed of $N_{c}=969$ cells, where many polygons have holes. To each cell we associate a $2$-order signal $s_2^{cc}(k)$, given by the intensity $f_i$ of the signals on the corresponding domain,
 as illustrated 
in Fig. \ref{fig:Image_monna}.
Clearly, this complex is tuned to the image and is then better able to grasp more structured
relationships among the image pixels.
For example, a big part of the hair of Mona Lisa gives rise to a big polygonal cell  associated to  pixels having similar intensity. 
It is also useful to notice that it is quite likely that the adopted space partition yields hollow cells, like for example
the  cell $c_1^{2}$ that has a hole represented by cell $c_2^{2}$, as indicated by the two arrows in Fig. \ref{fig:Image_monna}.\\ 
It is now interesting to compare the two inferred simplicial and planar cell complexes, in terms of image representation. Specifically, using the second order  Laplacian matrices $\mL_2=\mU_{2} \boldsymbol{\Lambda}_{sc}  \mU_{2}^{ T}$ and  $\mL_2^{cc}=\mU_{2}^{cc} \boldsymbol{\Lambda}_{cc} \mU_{2}^{cc \, T}$, associated with the SC and CC complexes, respectively, we project the signals $\bs_2$ and $\bs_2^{cc}$
over the eigenvectors bases $\mU_2$ and $\mU_2^{cc}$, thus obtaining the following simplicial and cell Fourier coefficients:
\beq \label{eq_repr_Four}
\hat{\bs}_2=\mU_2^T \bs_2, \qquad \hat{\bs}_2^{cc}=\mU_2^{cc \, T} \bs_2^{cc}.
\eeq
Then, we  build the matrices (bases)   $\mU_{2,\mathcal{K}}$ and $\mU_{2,\mathcal{K}_{c}}^{cc}$  containing as columns the eigenvectors, respectively, with indexes in the sets $\mathcal{K}$ and $\mathcal{K}_c$, associated with   the highest  $K$ coefficients in (\ref{eq_repr_Four}).  Then, we get the  $K$-sparse signal representations 
$\hat{\bs}_{2,\mathcal{K}}=\mU_{2,\mathcal{K}}^T\bs_{2}$ and  $\hat{\bs}_{2,\mathcal{K}_c}^{cc}=\mU_{2,\mathcal{K}_c}^{cc \, T} \bs_2^{cc}$. Finally, the images are reconstructed as 
$$\bar{\bs}_2=\mU_{2,\mathcal{K}} \hat{\bs}_{2,\mathcal{K}}, \qquad  \bar{\bs}_2^{cc}=\mU_{2,\mathcal{K}_c}^{cc }\hat{\bs}_{2,\mathcal{K}_c}^{cc}.$$
In Fig. \ref{fig:MSE_Spars_monna}, we  plot the normalized mean squared  error in the reconstructed image vs. its sparsity, i.e. the  number $K$ of bases (eigenvectors) used to represent the signal.
\begin{figure}[t!]
\centering
\includegraphics[width=8.5cm,height=5.3cm]{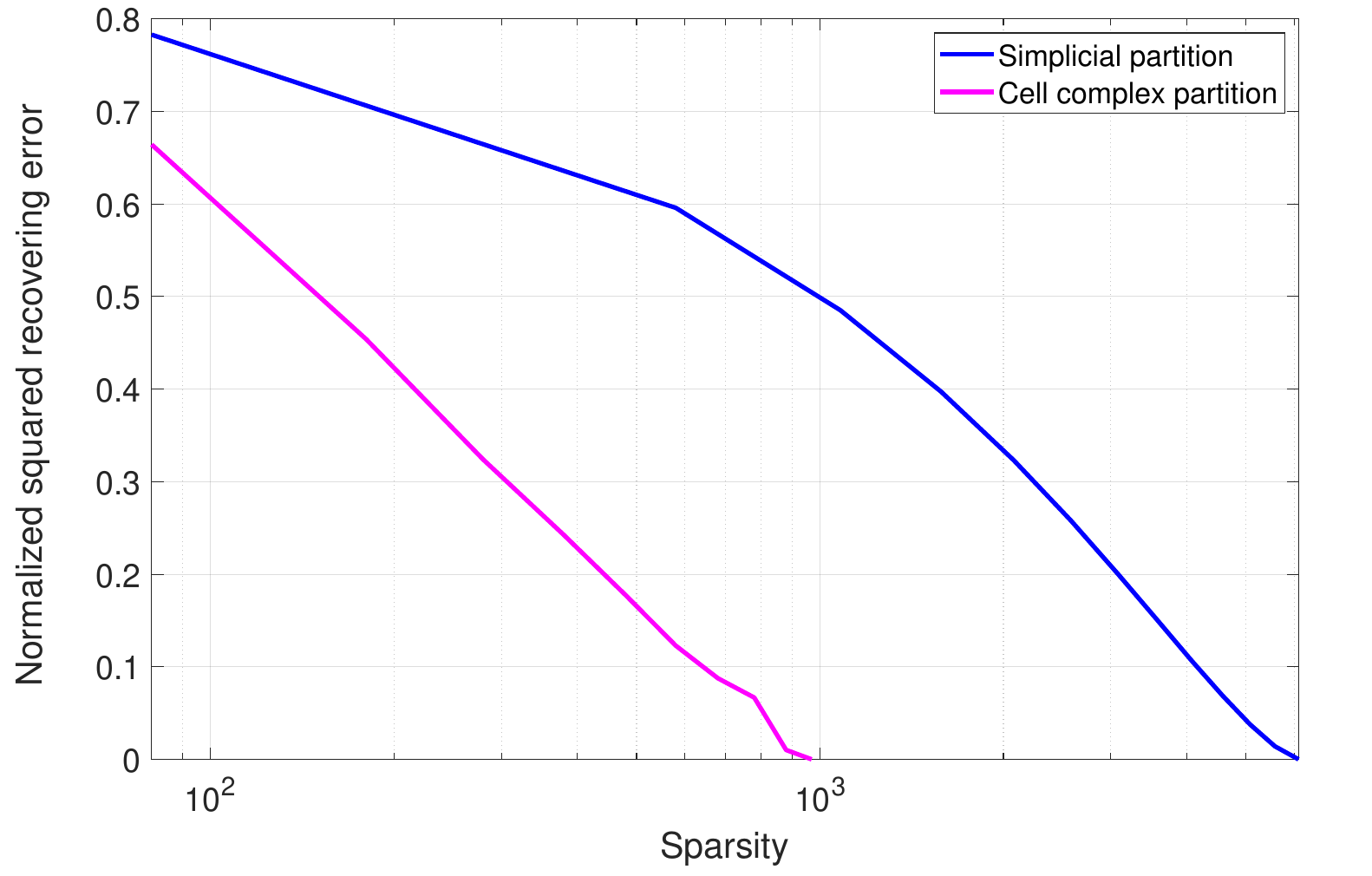}
\caption{Normalized  squared error vs. the signal sparsity.}
\label{fig:MSE_Spars_monna}
\end{figure}
It can be  noticed the considerable  gain in terms of accuracy/sparsity trade-off that a  cell based representation is able to provide with respect to a simplicial complex partition. Intuitively speaking, the advantage comes from a better tuning of the cell complex to the image.}

\section{Topology inference and sparse representation}
\label{Sec: Inference}
The goal of this section is to propose a method to infer the structure of a cell complex from data \textcolor{black}{by generalizing the method proposed in \cite{barb_2020} for the inference of simplicial complexes.} The underlying idea is that the inclusion of polygons of order greater than three yields a richer dictionary that can be exploited to find a better trade-off between sparsity and representation accuracy with respect to the simplicial  case.\\
The general method is inductive and it starts from a zero order set, i.e. an ensemble of points, and infers the structure of a first order set by exploiting the relations between the signals associated to the zero-order elements. This first step has been largely studied as the graph topology inference problem, see, e.g., \cite{mateos2019connecting}, 
\cite{Sar_Barb_Inf}. The successive steps are a generalization of the  first step, every time considering an order up. So, in general, the goal of the $n$-th step is to infer the structure of an $n$-th order cell complex from signals defined over the $(n-1)$-complex. We focus on the inference of the second order complex, but the procedure can be generalized to higher order structures in a straightforward manner.  Hence, in our case we start from an underlying graph, whose incidence matrix $\mB_1$ is then known, and we want to infer the structure of $\mB_2$ from the signals defined over the edges of the graph. We start from the observation of $M$ snapshots of edge signals $\bx^1(m)$, $m=1, \ldots, M$.
As a first step, we  check if it is really necessary to go beyond graphs and then introduce the matrix $\mB_2$ to represent the observed data.
Since the only signal components that may depend on $\mB_2$ are the solenoidal and harmonic components, we project the observed signals $\bx^1(m)$ for $m=1,\ldots,M$ onto the solenoidal and harmonic subspaces by computing
\begin{equation}
\label{proj_on-sh}
\bx^1_{\text{sH}}(m)=\left(\mathbf{I}-\mathbf{U}_{irr}\mathbf{U}_{irr}^T\right) \bx^1(m), \quad m=1,\ldots,M
\end{equation} where $\mU_{irr}$ is the matrix whose columns are the eigenvectors associated to the non-zero eigenvalues of $\mB_1^T \mB_1$.
If the \textcolor{black}{ratio between the overall energy of the $M$ signal components $\bx^1_{\text{sH}}(m)$ and that  of the signals $\bx^1(m)$ is lower than a given threshold}, we set $\mB_2=\mathbf{0}$ and stop, otherwise we proceed to estimate $\mB_2$.
To do that, we build a matrix $\mB_2$ as in \eqref{eq:B2_cell}, identifying all polygons in the graph.   This matrix can be rewritten as
\begin{equation}
\label{B2}    \mB_2\mB_2^T=\sum_{n=1}^{P}q_n \mathbf{b}_n \mathbf{b}_n^T
\end{equation}
where $\mathbf{b}_n$ is the column of $\mB_2$ associated with the $n$-th 2-cell (polygon), whose entries are all zero except the entries associated with the  edges of the $n$-th cell. \textcolor{black}{The coefficient $q_n \in \{0,1\}$ is introduced to select the columns of $\mB_2$}
for which the total circulation of the observed edge signal along  the corresponding $2$-cell is minimum. Note that the circulation of the generic component $\bx_{sH}^1(m)$ along the boundary of the $n$-th polygon is $\mb_n^T\bx_{sH}^1(m)$. If the \textcolor{black}{overall sum of the  norm }  of this circulation is sufficiently small over all set of data, we set $q_n=1$, i.e. we fill the $n$-th polygon, otherwise we set $q_n=0$ and leave the corresponding polygon empty.
Mathematically speaking, the inference problem boils down to solving
the following problem
\beq \label{eq:optm_T}
\begin{array}{lll}
\underset{\mathbf{q} \in \{0,1\}^{P}}{\min}
& f_{\text{TV}}(\mathbf{q}):=\ds\sum_{n=1}^{P} \sum_{m=1}^{M} q_{n}  [\mb_n^T\bx_{sH}^1(m)]^2  \quad (\mathcal{P}_{\text{MTV}}) \medskip\\
\quad \, \! \text{s.t.}
&  \parallel \mathbf{q} \parallel_0 = q^*, \quad  q_n \in \{0,1\}, \forall n,
\end{array}
\eeq
where $q^*$ is the number of (filled) $2$-cells. 
\textcolor{black}{Note that the objective function  $f_{\text{TV}}$  represents  the sum of the squared variations (curls) of the  edge signals  along  the $2$-cells of the complex.
}
Although problem $\mathcal{P}_{\text{MTV}}$ is non-convex, it can be easily solved in closed form, as in  \cite{barb_2020}. In fact, defining the nonnegative coefficients
$d_n=\sum_{m=1}^{M} [ \mb_n^T\bx_{sH}^{1}(m)]^2$, the optimal solution of problem $\mathcal{P}_{\text{MTV}}$ can be obtained by sorting the coefficients $d_n$ in increasing order and selecting the columns of $\mB_2$ corresponding to the indices of the  $q^{\star}$ smallest coefficients $d_n$.
Note that the proposed inference strategy can be extended to planar hollow cell complexes by using the incidence matrix $\mB_2$ in (41).\\
\textcolor{black}{\textbf{Remark.} In problem $\mathcal{P}_{\text{MTV}}$, we need  to select the optimal number $q^{*}$ of filled $2$-cells first. However, given the simplicity of the closed form solution,  we  can compute the solution associated to all  values of $q^{*}$ and then  select the  $q^{*}$ value yielding the best sparsity/accuracy trade-off.}\\
To check the effectiveness of the proposed approach, we compare the method with the corresponding method used in \cite{barb_2020} to infer a simplicial complex from data and with graph-based approaches. The criterion used for the comparison is the trade-off between sparsity of the representation and fitting error. 
To perform the comparison, given each observed flow vector $\bx^1$, we find the sparse vector $\bs^1$ as solution of the following basis pursuit problem
\cite{Donoho98}:
\beq \label{eq:bas_pur}
\begin{array}{lll}
 \underset{\textcolor{black}{{\hat{\bs}}^1} \in \mathbb{R}^E}{\text{min}} & \parallel
\hat{{\bs}}^1\parallel_1   \qquad \qquad \qquad (\mathcal{P}_B)\\
 \; \; \text{s.t.} & \parallel
 {\bx}^1 -\mV \hat{{\bs}}^1\parallel_F \leq \epsilon
 \end{array}
\eeq
where $\mV$ is a dictionary matrix, whose columns are the eigenvectors $\mU_1$ of the Laplacian matrices built using four alternative methods: 1) the cell-based method, where $\mL_1=\mB_1^T \mB_1+\mB_2 \mB_2^T$, with $\mB_2$  built using the cell-based method described above; 2) the simplicial-based method, where $\mL_1=\mB_1^T \mB_1+\mB_2 \mB_2^T$, with $\mB_2$  built using the simplicial-based method proposed in \cite{barb_2020}; 3) the graph-based method, where $\mL_{1,d}=\mB_1^T \mB_1$ incorporates only a purely graph-based representation; 4) a line graph-based method, where the line graph \textcolor{black}{\cite{godsil2001algebraic}} of a graph $\mathcal{G}$ is a graph whose nodes correspond to the edges of $\mathcal{G}$ and its Laplacian matrix is defined as $\mL_{LG}=\text{diag}(\mA_{LG}\mathbf{1})-\mA_{LG}$ with $\mA_{LG}=|\mL_{1,d}-2\mI|$. \textcolor{black}{ By varying  the  positive coefficient $\epsilon$ in (\ref{eq:bas_pur}), we can explore the trade-off between sparsity and accuracy.} \\

 \noindent \textbf{Data Traffic Network.}
We tested the above four methods on the real traffic data measurements on the German National Research and Education Network  operated
by the German DFN-Verein (DFN) \cite{OrlowskiPioroTomaszewskiWessaely2010}.
The
backbone network consists of $50$ nodes, $89$ links and $39$ potential $2$-cells.
The data traffic is aggregated daily over the month of February 2005 and the data  measurements, expressed in Mbit/sec, are observed over the links of the backbone network.
In Fig.  \ref{fig:spara_DFN} we draw the sparsity vs. the representation error $\parallel \bx^1-\mV \bs^1 \parallel_F$ obtained with the above four methods by varying the  coefficient $\epsilon$. We can see from Fig.  \ref{fig:spara_DFN} that the cell-based approaches outperforms all other methods in terms of sparsity/accuracy trade-off.
\begin{figure}[t]
\centering
\includegraphics[width=7.3cm,height=5.0cm]{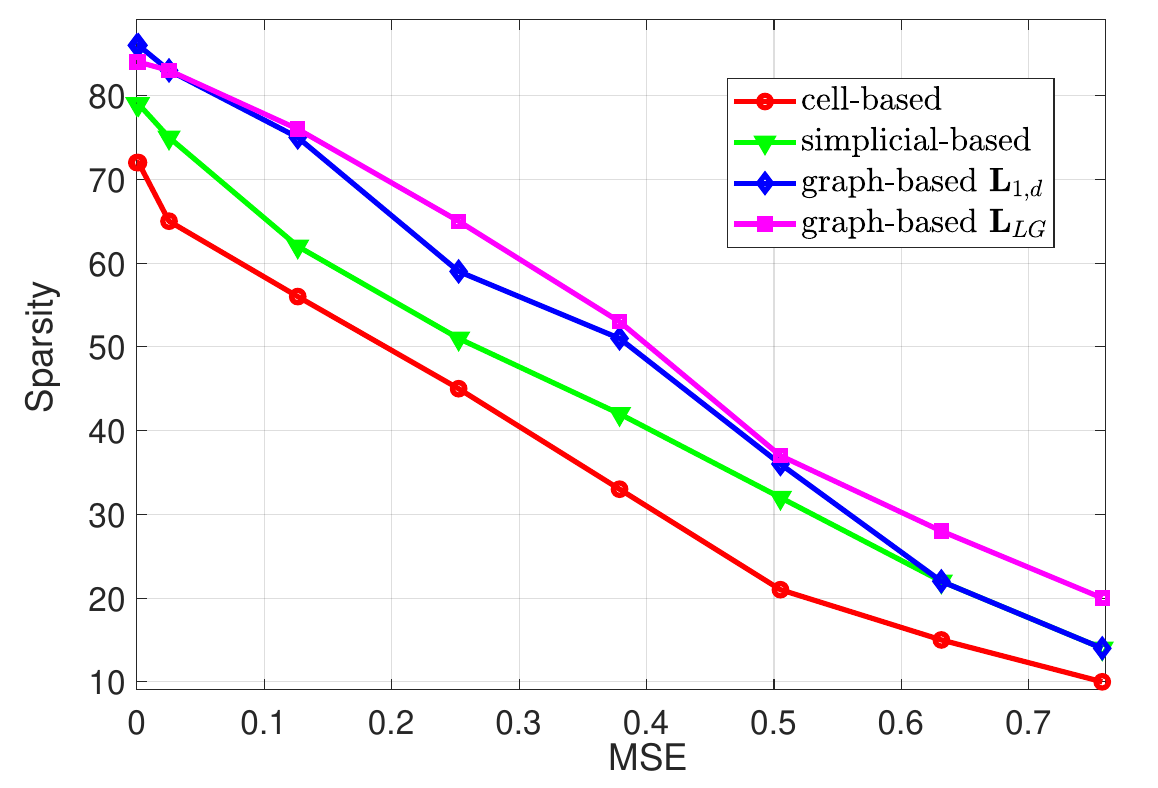}
\caption{Sparsity versus mean squared error.}
\label{fig:spara_DFN}
\end{figure}

 \section{Extraction of  space invariants (homologies)}
\label{sec: Filter_harm}
\textcolor{black}{If embedded in a two-dimensional (2D) real domain $\mathbb{R}^2$, a geometric cell complex typically fills a surface. More generally, in a higher dimension domain $\mathbb{R}^m$, the complex may fill an $m$-th dimensional volume. A basic result of algebraic topology is that some of the fundamental invariants of a cell complex are the number of connected components and the number  of holes (cavities) appearing at increasing dimensions \cite{munkres2000topology}. When analyzing flows along the edges of a complex, it is of particular interest to identify closed paths, or cycles, and check the flow of signals along these cycles. More specifically, looking for invariants, it is useful to distinguish between two kinds of cycles: the cycles, also named boundaries, that enclose a filled surface (or volume) and the  cycles that enclose a hole (or cavity). Detecting cycles surrounding holes is indeed a fundamental step to identify invariants of the structure and identifying invariants is a key step in devising effective geometric deep learning algorithms \cite{bronstein2021geometric}.  Detecting cycles in flow signals finds applications in protein-protein interaction networks \cite{davis2015topology}  to recognize  proteins with similar biological functions or to categorize protein into enzyme and not-enzyme structures
\cite{borgwardt2005protein}. The distinction between cycles that are or not boundaries can also be made using simplicial complexes, but in the latter case, for example, we cannot include filled polygons of dimension greater than three
and in many applications such a restriction may prevent us from finding cycles with minimal length \cite{bodnar2021weisfeilercell}, \cite{horak2009persistent}.
}
The harmonic (edge) signals are (global) circulations that are not representable as a linear combination of curl flows. The analysis of harmonic signals over simplicial complexes has been used to detect global inconsistency in statistical ranking  \cite{jiang2011},  to localize coverage holes in a sensor network represented as a Rips complex  \cite{tahbaz2010distributed}, and  to identify topologically similar trajectories through embedding  into the harmonic subspace \cite{SCHAUB2021}. The harmonic signals can be represented as a linear combination of a basis of eigenvectors spanning the kernel of $\mL_1$. \textcolor{black}{Since there is no unique
way to identify a basis for such a subspace, we look for the sparsest harmonic  cycles, i.e. cycles that do not include boundaries of $2$-order cells, and that  cover a minimum number of edges around holes.}
More specifically, in this section, we propose a distributed optimization strategy to identify holes over cell complexes from noisy measurements by finding the optimal trade-off between data fitting error and sparsity of the edge signals. We derive the sparsest harmonic component by  using an iterative subspace projection optimization method \cite{di2020distributed} converging to the desired  harmonic solution and amenable for  a distributed implementation.
Then, if our goal is detecting holes we have to remove from the harmonic signal $\bs_H$  the curl components that can be expressed as  combinations of  circulations along the boundaries of the $2$-cells.
Assume that the noisy edge  signal $\bx^{1} \in \mathbb{R}^E$  is observed over the edges of a cell complex of order $2$ with 
\beq \bx^{1}=\bs_{H}+\bn\eeq
where  $\bn$ is an additive noise vector.
Our goal is to filter  the harmonic component $\bs_H$ in order to obtain a sparse representation useful to detect holes by using a distributed algorithm. \textcolor{black}{Then, we have to extract from the observed edge signal $\bx^{1}$ the harmonic component while removing  the contributions corresponding to curl flows of $2$-cells,  represented by the signal  $\mathbf{s}_{sol}=\mathbf{B}_2 \mathbf{s}_2$. 
 This entails minimizing both the fitting error and the $l_1$-norm of the vector $\bz=\bs_H+\mB_2 \bs^2$, with the goal of identifying the sparsest edge signal ($1$-cycle) $\bz$ belonging to the kernel of $\mB_1$. This signal  is homologous to $\bs_H$ and tightly encircles the holes of the complex.}
Hence, we find the optimal filtered harmonic signal $\bs_H$ and the second-order cell signal $\bs^2$ solving the following optimization problem
\beq
\begin{array}{lll}
&\underset{\bs_H, \bs^2}{\min} \quad & \parallel \bx^{1}-\bs_H \parallel^2 + \, \gamma \parallel \bs_H+\mB_2 \bs^2 \parallel_1  \qquad (\mathcal{Q}) \\
&\text{s.t.} & \bs_H= \mW_h \bs_H\\
\end{array}
\eeq
where $\mW_h=\mI-\epsilon \mL_1$ and  $\gamma$ is a non-negative coefficient controlling
the trade-off between fitting error and sparsity of the recovered edge signal $\bz=\bs_H+\mB_2 \bs^2$. The constraint $\bs_H=\mW_h \bs_H$
forces $\bs_H$ to belong to the harmonic subspace \cite{di2020distributed}. To
ensure convergence to the harmonic solution of the ensuing iterative solution of the above problem,  we choose  $0<\epsilon<2/\lambda_{max}(\mL_1)$
where $\lambda_{max}(\mL_1)$ is the maximum eigenvalue of $\mL_1$.
A distributed solution of problem $\mathcal{Q}$ can be obtained proceeding as in  
\cite{di2020distributed}, by converting the constrained problem $\mathcal{Q}$ into a sequence of unconstrained problems that, at each step $k$, are formulated as follows: \vspace{-0.01cm}
\beq \label{eq:probl_uncon}
\begin{split}
\underset{\bs_H, \bs^2}{\min} \; \; f_h(\bs_H,\bs^2)&:=\parallel \bx^{1}-\bs_H \parallel ^2 + \,\gamma \parallel \bs_H+\mB_2 \bs^2 \parallel_1 \\ &+ \frac{1}{2 \mu[k]}  \bs_H^T(\mI-\mW_h) \bs_H
\end{split}
\eeq
where $\{\mu[k]\}_k$ is a positive non-increasing sequence forcing  the signal $\bs_H$ to belong to $\text{ker}(\mI-\mW_h)$ as $k \rightarrow \infty$.
\textcolor{black}{Since we chose  $\epsilon>0$ such that $\mI-\mW_h=\epsilon \mL_1$, from the positive semidefiniteness  of $\mL_1$, it follows that
the matrix $\mI-\mW_h$ is  positive semidefinite.} 
\begin{algorithm}[t!]
\small
      {Data: $\bx^1$}.  {Start} {with} {random} {vectors} $\bs_{H}[0] \in \mathbb{R}^{E\times 1}$,
     $\bs^{2}[0]\in \mathbb{R}^{T\times 1}$, $a=1.5e^{-3}$, $k=0$, $\mu[k]=a$. \medskip

    \quad {For} {each} {time} $k > 0$, {repeat}    \medskip

      \quad \quad  1)  $\bs^2[k+1]=\bs^2[k]-\mu[k]\partial_{\bs^2} f_h (\bs_H[k],\bs^2[k])$ \medskip

      \quad \quad 2) $\bs_H[k+1]=\bs_H[k]-\mu[k] \partial_{\bs_H} f_h(\bs_H[k],\bs^2[k])$ \medskip

       \quad \quad 3) $\mu[k+1]=\ds \frac{a}{\sqrt{k+1}}$\medskip

       \quad \quad 4) $k=k+1$
   \quad
   \caption{\small{Distributed Optimization of the Harmonic Filtering}}
 \label{algorithm:Alg_1}
\end{algorithm}
Then, problem (\ref{eq:probl_uncon}) is a convex non-differentiable problem that  can be solved using a (sub)gradient method  where the following recursive rules are updated:
\beq  \label{eq:up_grad}
\begin{split}
& \bs^2[k+1]=\bs^2[k]-\mu[k] \partial_{\bs^2} f_h(\bs_H[k],\bs^2[k])\\
& \bs_H[k+1]=\bs_H[k]-\mu[k] \partial_{\bs_H} f_h(\bs_H[k],\bs^2[k])
\end{split}
\eeq
with
\beq \label{eq:grad}
\begin{split}
&\partial_{\bs^2} f_h(\bs_H[k],\bs^2[k])= \gamma \, \mB_2^T \text{sign}(\bs_H[k]+\mB_2 \bs^2[k])\\
& \partial_{\bs_H} f_h(\bs_H[k],\bs^2[k])\!=\!-2 (\bx^1-\bs_H)\!+\!\gamma \,\text{sign}(\bs_H[k]+\mB_2 \bs^2[k]).
\end{split}
\eeq
Note that the updating rules in (\ref{eq:up_grad}) are amenable for a distributed implementation since the subgradients in
(\ref{eq:grad}) can be derived combining $2$-  or $1$-order signals according to the neighborhood relationships captured by the incidence matrix $\mB_2$. The final distributed algorithm is illustrated in Algorithm $1$ where, to ensure convergence,  we adopt the nonsummable diminishing step size rule $\mu[k]=a/\sqrt{k}$ with $a>0$.\\
\begin{figure}[t!]
\centering
\includegraphics[width=8.0cm,height=4.1cm]{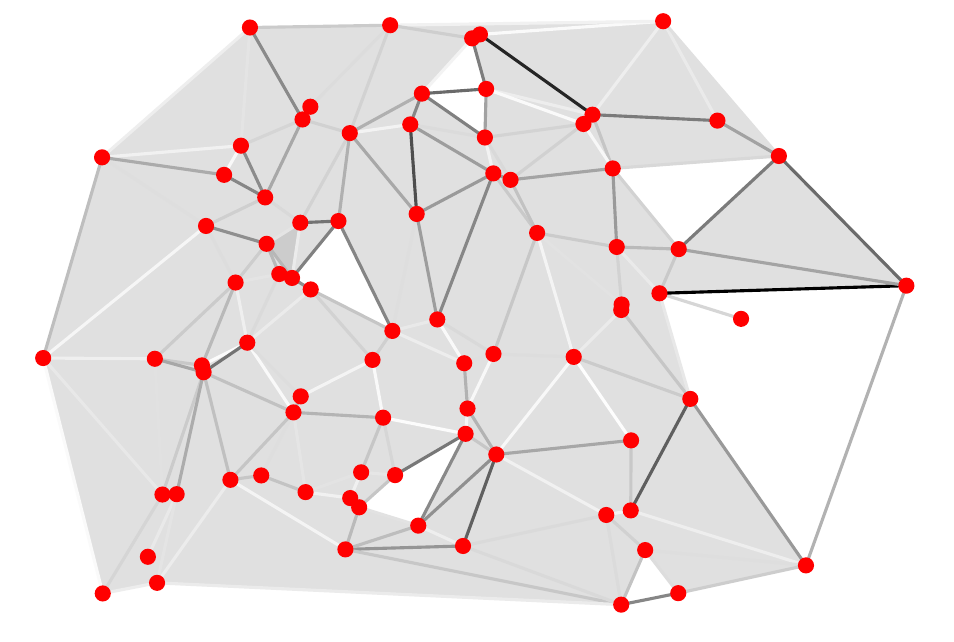}\\
(a) cell complex
\includegraphics[width=5.8cm,height=8.4cm,angle=-90]{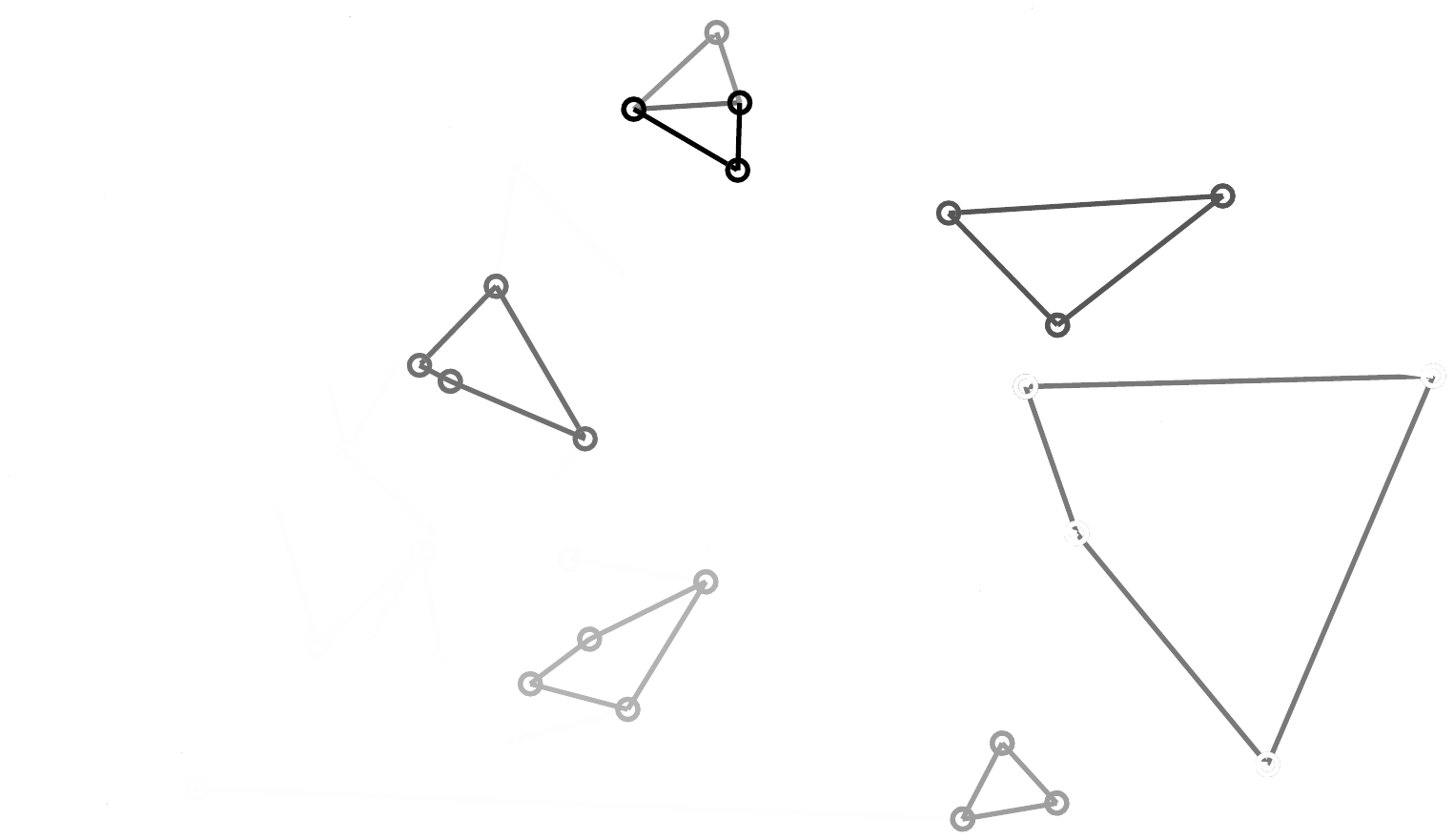}\vspace{-0.2cm}
(b) recovered edge signal
\caption{Filtering harmonic component.}
\label{fig:rec_sign}
\end{figure}
To test the goodness of the proposed harmonic filtering method in identifying minimal-length cycles around holes, we consider the cell complex in Fig. \ref{fig:rec_sign}(a)  composed of $N=80$ nodes, a number of edges $E=188$ and $102$ polyhedral $2$-cells. The covered surface contains $7$ holes \textcolor{black}{so that} the dimension of the harmonic subspace is equal to $7$.
The observed edge signal is  generated as $\bx^{1}=\mU_H \hat{\bs}^1 +\bn$ where $\mU_H$ are the eigenvectors spanning the kernel of $\mL_1$, $\hat{\bs}^1$ is a random vector with entries uniformly distributed between $[0,1]$ and the vector $\bn$  has  zero-mean Gaussian entries with variance $\sigma_n^2=0.05$. The signal values are encoded in the gray color associated with the links in  Fig. \ref{fig:rec_sign}(a).
The sparse edge signal $\bs_H+\mB_2 \bs^2$  recovered by running   Algorithm $1$ is reported in Fig.
\ref{fig:rec_sign}(b), by  setting the trade-off coefficient $\gamma=50$.
We can observe that the algorithm converges to a harmonic component that is non-zero only along the cycles surrounding the holes. \textcolor{black}{This is an example testifying the goodness of the method to identify the (shortest) cycles of the complex.}

\section{Filtering over cell complexes}
The design of filters operating over signals 
defined over the edges of a graph, exploiting simplicial complex topologies, has been considered  recently in \cite{barb_2020}, \cite{ebli2020simplicial}, \cite{Yang2021}. 
\textcolor{black}{A  first introduction to filtering over cell-complexes has been provided in \cite{grady2010}.} 
In \cite{Sar_Bar_Tes} we introduced some preliminary results on finite impulse response (FIR) filters over cell complexes.
\textcolor{black}{A local} filter operating over an edge signal vector $\bs^1$ can be modeled as
\begin{equation} \label{eq:Lin_fil}
    \by^1=\mH\,\bs^1
\end{equation}
where we assume the support of $\mH$  given by the (upper and lower) neighbors of each edge.
\textcolor{black}{
Generalizing the FIR filter design operating over graphs, see e.g. 
\cite{Segarra_17}, to cell complexes, a polynomial FIR filter based on the first order Laplacian  has  the following form
\beq \label{eq: filter_sep}
\mH= \sum_{k=0}^{L} a_k \mL^{k}_{1}= \alpha_0 \mI+\sum_{k=1}^{L} a_k^I \mL^{k}_{1,d}+ \sum_{k=1}^{L} a_k^s \mL^{k}_{1,u}
\eeq
where  $\alpha_0:=a_0^I+a_0^s$
and $L$ is the filter order. In the last equality, we used \eqref{L1up+L1down} and the property that, since $\mB_1 \mB_2=\mathbf{0}$, we can always write $\mL_1^k=\mL_{1,d}^k+\mL_{1,u}^k$.
The above expression shows that filtering edge signals occurs by combining, separately, along lower neighbors (first term) and upper neighbors (second term). 
In (\ref{eq: filter_sep}),} 
we  adopt two sets of coefficients $\{a_k^I\}_{k=1}^{L}$ and $\{a_k^s\}_{k=1}^{L}$ to design the filters operating along the  lower and upper neighbors. This formulation was already proposed in \cite{Yang2021} to design filters operating over simplicial complexes.\\
A possible way to select the two sets of filter coefficients 
can be obtained by imposing two spectral masks on the  irrotational and  solenoidal components.
Applying a CFT to the output of the filter,  we can write the spectrum of the output as
\beq \label{eq:FIR_sch}
\begin{split}
\hat{\by}^1&=\mU_1^T \by^1
=\left(\alpha_0 \mI + \sum_{k=1}^{L} a_k^I \boldsymbol{\Lambda}_{1,d}^k + \sum_{k=1}^{L} a_k^s \boldsymbol{\Lambda}_{1,u}^k\right) \hat{\bs}^1,
\end{split}
\eeq
where $\mU_1$ are the eigenvectors of $\mL_1$ and $\boldsymbol{\Lambda}_{1,d}$, $\boldsymbol{\Lambda}_{1,u}$ are  the diagonal matrices with entries the  eigenvalues $\lambda_d(i)$, $\lambda_u(i)$ of, respectively, the lower and upper Laplacians.
From this expression we can notice that the effect of the filters 
is to impose on the irrotational and solenoidal frequencies the spectral masks
\beq
h^I(\boldsymbol{\Lambda}_{1,d})=\sum_{k=1}^{L} a_k^I \boldsymbol{\Lambda}_{1,d}^k, \quad h^s(\boldsymbol{\Lambda}_{1,u})=\sum_{k=1}^{L} a_k^s \boldsymbol{\Lambda}_{1,u}^k,
\eeq
whose diagonal entries, given by respectively, $h^I(\lambda_d(i))=\sum_{k=1}^{L} a_k^I {\lambda}_{d}(i)^k$ and $h^s(\lambda_u(i))=\sum_{k=1}^{L} a_k^s {\lambda}_{u}(i)^k$ are the frequency responses of the filters. Note that, if the Laplacian $\mL_1$ contains eigenvalues of multiplicity greater than $1$, \textcolor{black}{ all associated to the same subspace, i.e. to the irrotational, solenoidal or harmonic subspace, the spectral filter is unable to distinguish components belonging to the same subspace. 
However, if  a  multiple  eigenvalue is   associated with both  solenoidal and irrotational  eigenvectors, the spectral filter can  distinguish components within either the solenoidal or irrotational subspace, only if the    multiplicity of the eigenvalue within the subspace  is equal to one.} A relevant example of this situation occurs when the second order cell complex covers a surface, but it leaves a number of holes $N_h>0$. In such a case, the multiplicity of the null eigenvalue is exactly $N_h$ \cite{munkres2018elements}. From \eqref{eq:FIR_sch}, we can immediately see that the harmonic component, associated with the null eigenvalue, is simply multiplied by a scalar coefficient $\alpha_0$. Hence, the filter structure in \eqref{eq:FIR_sch} is not adequate to filter {\it within}  the harmonic subspace. To overcome this issue, we can resort to method proposed in the previous section that operates explicitly on the harmonic component. 

\noindent \textbf{Filtering the solenoidal and irrotational components.}\\
In this section,
we will provide theoretical conditions under which the design of
two separate FIR filters based on the lower and upper Laplacians  may offer performance gains with respect to a  single FIR filter handling jointly the solenoidal and irrotational signals.
The solenoidal and irrotational components refer, by construction, to eigenvalues different from zero, then we define the corresponding filters by setting $\alpha_0=0$ in \eqref{eq: filter_sep} and \eqref{eq:FIR_sch}.\\
Given the spectral mask  $h(\lambda_{1,i})$ that we wish to implement over the non-zero eigenvalues $\lambda_{1,i}$ of $\mL_1$, 
the spectral $\mL_1$-based filter   can be written as
 \beq
 \label{eq:Fir1}
 \mH=\overline{\mU}_{1} h(\boldsymbol{\Lambda}_1)\overline{\mU}_1^T
 \eeq
where $\overline{\mU}_1$ contains, in its columns, only the eigenvectors associated to the non-null eigenvalues of $\mL_1$, whereas $h(\boldsymbol{\Lambda}_1)$ is a diagonal matrix with entries $h(\lambda_{1,i})$, $i=1,\ldots,n_t$, and $n_t=\text{rank}(\mL_1)$. 
The filter in (\ref{eq:Fir1}) can be decomposed in the sum of two orthogonal filters associated with the irrotational and  solenoidal subspaces, i.e.
\beq \label{eq:filter_sol_irr}
\mH=\mH_{d}+\mH_{u}=\mU_{d} h^I(\boldsymbol{\Lambda}_{d}) \mU_{d}^T+\mU_{u} h^s(\boldsymbol{\Lambda}_{u}) \mU_{u}^T,
\eeq
where $\boldsymbol{\Lambda}_{d}$ and $\boldsymbol{\Lambda}_{u}$ are the diagonal matrices whose entries are the non-zero eigenvalues of $\mL_{1,d}$ and $\mL_{1,u}$, whereas
$\mU_{d}$ and $\mU_{u}$ are the matrices containing  the associated  eigenvectors. Let us now consider the  two cases where the solenoidal and irrotational  filters are designed using a common set of coefficients or two separate sets.
 \subsubsection{Common filtering}
Let us first approximate the filter $\mH$  using a FIR filter  based on the first order Laplacian given by
\vspace{-0.2cm}
 \beq \widehat{\mH}=\sum_{k=1}^{L} a_k \mL_1^k.
 \eeq
The filter coefficient vector $\ba=[a_1,\ldots,a_L]^T$ can be found as the solution of a least squares problem.
Let us define the matrices
 \beq \label{eq:Phi_I_S}
 \boldsymbol{\Phi}_{I}=[\boldsymbol{\lambda}_I, \boldsymbol{\lambda}_I^2,\ldots, \boldsymbol{\lambda}_I^L], \;  \boldsymbol{\Phi}_{s}=[\boldsymbol{\lambda}_s, \boldsymbol{\lambda}_s^2,\ldots, \boldsymbol{\lambda}_s^L] \eeq with $\boldsymbol{\lambda}_I$ and $\boldsymbol{\lambda}_s$ the column vectors associated with  the diagonal entries of $\boldsymbol{\Lambda}_{d}$ and $\boldsymbol{\Lambda}_{u}$.
Then, using the equalities in (\ref{eq:filter_sol_irr}) and the orthogonality condition $\mathbf{U}_{d}^T \mathbf{U}_{u}=\mathbf{0}$, we get
\vspace{-0.2cm}
\beq \label{eq: sum_conv2}
\begin{split}
\parallel \mathbf{H}-\ds \sum_{k=1}^{L} a_k\, \mL_1^k  \parallel^2_F &= \parallel \mathbf{h}^{s}-\boldsymbol{\Phi}_{s}  \ba \parallel^2_F+ \parallel \mathbf{h}^{I}-\boldsymbol{\Phi}_{I}  \ba \parallel^2_F\\&=\parallel \mathbf{h}-\boldsymbol{\Phi}_{t}  \ba \parallel^2_F
\end{split}
\eeq
where $\mathbf{h}^{s}=\text{diag}(h^s(\boldsymbol{\Lambda}_{u}))$, $\mathbf{h}^{I}=\text{diag}(h^I(\boldsymbol{\Lambda}_{d}))$, $\mathbf{h}=[\mathbf{h}^{s};\mathbf{h}^{I}]$ and $\boldsymbol{\Phi}_{t}=[\boldsymbol{\Phi}_{s};\boldsymbol{\Phi}_{I}]$.
Hence, in the spectral domain the least squares problem we need to solve is 
\beq \label{eq:H_g1}
\underset{\ba \in \mathbb{R}^{L}}{\min} \quad f(\ba):=\parallel \mathbf{h}-\boldsymbol{\Phi}_t {\ba} \parallel^2_F \eeq
whose   closed form solution is given by
$
\ba=\boldsymbol{\Phi}^{\dag}_t \mathbf{h}
$
where $\boldsymbol{\Phi}^{\dag}_t$ is the Moore-Penrose pseudo-inverse of $\boldsymbol{\Phi}_t$.

\subsubsection{Independent filter design}
The FIR filter design illustrated above is the result of a parsimonious choice of the number of filter coefficients, but may be often unable to well approximate different masks for
the solenoidal and irrotational components. To improve the situation, recalling \eqref{eq: filter_sep} and \eqref{eq:FIR_sch}, it is worth to  design the two FIR filters separately, at the slightly increased cost of using two sets of coefficients. More specifically,
given the linear filters $\mH_{u}$ and $\mH_{d}$, our goal now is to find the optimal FIR filters $\hat{\mH}_{u}$, $\hat{\mH}_{d}$ approximating them such that
\beq \label{eq:filter_s_I}
\hat{\mH}_{u}=\ds \sum_{k=1}^{L} a_k^{s} \mL_{1,u}^k, \quad \hat{\mH}_{d}=\ds \sum_{k=1}^{L} a_k^{I} \mL_{1,d}^k.
\eeq
Proceeding exactly as before,
the two filters can be found  by  solving the following least squares problems
 in the spectral domain:
\beq \label{eq:H_sol1}
\underset{\mathbf{a}^{s} \in \mathbb{R}^{L}}{\min} \qquad f_s(\ba^s):=\parallel \mathbf{h}^s-\boldsymbol{\Phi}_{s}  \ba^{s} \parallel^2_F,
\eeq
\beq \label{eq:H_irr1}
\underset{\mathbf{a}^{I} \in \mathbb{R}^{L}}{\min} \qquad f_I(\ba^I):=\parallel \mathbf{h}^{I}-\boldsymbol{\Phi}_{I}  \ba^{I}\parallel^2_F.
\eeq
The optimal coefficients solving  problems (\ref{eq:H_sol1}) and (\ref{eq:H_irr1})
can be derived in closed form.\\
Given the two optimization problems in (\ref{eq:H_sol1}) and (\ref{eq:H_irr1}), let us denote with $L$, $n_d$ and $n_{u}$, respectively, the filter length and the number of distinct non zero eigenvalues of $\mL_d^1$ and $\mL_{u}^1$. Then, as  proved in  Appendix B, it holds:
\begin{itemize}
    \item[a)] the optimal filter coefficients solving  (\ref{eq:H_sol1}), (\ref{eq:H_irr1}) are
    \beq \label{eq:sol_opt}
\ba^{s \star}=\boldsymbol{\Phi}_{s}^{\dag} \mathbf{h}^s, \qquad \ba^{I \star}=\boldsymbol{\Phi}_{I}^{\dag} \mathbf{h}^{I}.
\eeq
\item[b)] if $L\geq n_d$ and $L\geq n_{u}$, then
the pseudoinverse solutions in (\ref{eq:sol_opt})
are such that $f_s(\ba^{s \star})=f_I(\ba^{I \star})=0$. Furthermore, if $L=n_u$ and/or $L=n_{d}$ then the solutions of  (\ref{eq:H_sol1}) and/or (\ref{eq:H_irr1}) are unique.
\end{itemize}
\textcolor{black}{Even though  the independent filters design in general leads  performance gains with respect to a single filter due to the higher degree of freedoms, we indeed observed that there are some particular cases where this improvement does not occur. To better understand when this happens, in the following theorem we state the theoretical conditions under which  the design of independent FIR filters for the solenoidal and irrotational signals  attains better performance with respect to a single $\mL_1$-based  filter.} 
 Let us define
 \beq \label{eq:Phi_d_u}
\begin{split}
& \boldsymbol{\Phi}_{d}\!= \![\text{vec}(\mL_{1,d}^1),..., \text{vec}(\mL_{1,d}^L)], \boldsymbol{\Phi}_{u}\!= \![\text{vec}(\mL_{1,u}^1),..., \text{vec}(\mL_{1,u}^L)].
\end{split}
\eeq
Denote with $\mathcal{R}(\bA)$ the span of the columns of the matrix $\bA$, and with
$\mathbf{P}_{\mathcal{R}(\mathbf{A})}=\mathbf{A} \mathbf{A}^{\dag}$, $\mathbf{P}_{\mathcal{R}(\mathbf{A})^{\perp}}=\mathbf{I}-\mathbf{A} \mathbf{A}^{\dag}$
the projections, respectively,  onto $\mathcal{R}(\bA)$ and  onto its  complementary subspace.

\begin{proposition}
\textit{
Let  $f(\mathbf{a}^{\star})$, $f_s(\mathbf{a}^{s \, \star})$ and $f_I(\mathbf{a}^{I \, \star})$ be the  optimal solutions of the convex optimization problems in (\ref{eq:H_g1}), (\ref{eq:H_sol1}) and (\ref{eq:H_irr1}), respectively. Define the vector $\by=(\mathbf{P}_{\mathcal{R}(\boldsymbol{\Phi}_{d})}+\mathbf{P}_{\mathcal{R}(\boldsymbol{\Phi}_{u})}) \text{vec}(\mH)$. Then, it  holds:\\
a) if the system $\boldsymbol{\Phi}_{t}  \ba =\mathbf{h}$ is inconsistent and the matrices $\boldsymbol{\Phi}_{d}$ and $\boldsymbol{\Phi}_{u}$ are full column rank,  then:
\begin{itemize}
    \item[\text{i)}]  If $\mathbf{P}_{\mathcal{R}(\boldsymbol{\Phi}_{d}+\boldsymbol{\Phi}_{u})^{\perp}} \by=\mathbf{0}$ the disjoint and joint filtering strategies are equivalent, since we have:
     \beq f_I(\mathbf{a}^{I \, \star})+f_s(\mathbf{a}^{s \, \star}) = f(\mathbf{a}^{\star}); \eeq
     \item[\text{ii)}] If $\mathbf{P}_{\mathcal{R}(\boldsymbol{\Phi}_{d}+\boldsymbol{\Phi}_{u})^{\perp}}\by \neq \mathbf{0}$, the disjoint filtering strategy leads  smaller filter estimation errors  since it holds:
     \beq f_I(\mathbf{a}^{I \, \star})+f_s(\mathbf{a}^{s \, \star}) < f(\mathbf{a}^{\star});
\eeq
\end{itemize}
b) if the system $\boldsymbol{\Phi}_{t}  \ba =\mathbf{h}$ is consistent, then the disjoint and joint filtering strategies are equivalent  since we get
\beq
f_I(\mathbf{a}^{I \, \star})=f_s(\mathbf{a}^{s \, \star})= f(\mathbf{a}^{\star})=0.
\eeq}
\end{proposition}
\begin{proof}
See  Appendix C.
\end{proof}

From Prop. $1$, we can state  that, in terms of estimation error, the advantage achievable using the separate design depends on the structure of the desired filter $\mH$.
More specifically, by keeping  the length $L$ of the approximating FIR filters  fixed,  
if the system $\boldsymbol{\Phi}_t\ba=\mathbf{h}$ is consistent, then the common or separate design would provide the same accuracy.
If, instead, the system is inconsistent and the matrices $\boldsymbol{\Phi}_d$
and $\boldsymbol{\Phi}_u$ are full column rank, then we compute first the projection $\by$ of $\text{vec}(\mH)$ onto the subspace $\mathcal{R}(\boldsymbol{\Phi}_d)+\mathcal{R}(\boldsymbol{\Phi}_u)$. If $\by$
has a component falling outside the subspace $\mathcal{R}(\boldsymbol{\Phi}_d+\boldsymbol{\Phi}_u)$, then the separate design of the  solenoidal and irrotational parts yields a lower estimation error.\\
As a first numerical test, we compare the performance of  our filtering strategy with the approach proposed in \cite{Yang2021}, 
adapted to filtering signals defined over cell complexes by simply replacing $\mB_2$ with the matrix in  (\ref{eq:B2_cell}) capturing the upper incidence relations of a cell complex. Our optimal FIR filters are derived by solving the least squares problem in  (\ref{eq:H_sol1}) and (\ref{eq:H_irr1}), where we set $\alpha_0=0$. Indeed, albeit the difference between our approach and that of   \cite{Yang2021} is minimum, because we simply set $\alpha_0=0$ a priori, while in \cite{Yang2021} the choice of $\alpha_0$ is part of the optimization,  we will show that there is a clear advantage in terms of  accuracy. As a  numerical example, we consider the filtering of  the solenoidal component over a cell complex composed of $N=80$ vertices, $E=188$ edges and $96$  cells of order $2$ with $\text{dim}(\text{ker}(\mL_1))=13$.
The input signal $\bs^{1}$  is the sum of a solenoidal plus a harmonic signal, generated as $\bs^1_{sol}=\mU_{u} \bx_{sol}$ and $\bs^1_{H}=\mU_{H} \bx_{H}$, respectively,
where $\mU_{H}$ is a set of orthonormal vectors spanning the harmonic subspace and the vectors $\bx_{sol}, \bx_{H}$ have random i.i.d.  entries with uniform p.d.f. between $[0,a]$. To check the impact of the harmonic component, we keep the variance of the  solenoidal part  $\sigma^2_{sol}$ fixed, by setting $a=10$, and we increase the variance $\sigma_{H}^2$ of the harmonic component.
The frequency response of the solenoidal filter is chosen equal to
\begin{equation}
\label{h^s}
h^s(\boldsymbol{\lambda}_s; \alpha_s, b_s)=\frac{\beta_s}{1+ \exp(\alpha_s (\boldsymbol{\lambda}_s-b_s))}
\end{equation}
with $\beta_s=0.8$, $\alpha_s=b_s=3$, and the length of the FIR filters is  $L=11$.
In Fig. \ref{fig:conf_isufi}
we plot the average ratio $\frac{\parallel \hat{\mH}_u \bs^1-\mH_u \bs^{1}_{sol} \parallel}{\parallel \mH_u \bs^{1}_{sol} \parallel}$, i.e. the  normalized mean squared error between the signal at the output  of the solenoidal FIR filter taking as  input $\bs^1$ and the output of the ideal filter $\mH_{u}$, built as in \eqref{eq:filter_sol_irr} using the spectral mask in \eqref{h^s}, with input signal $\bs^1_{sol}$, versus the ratio $\sigma_H^2/\sigma_{sol}^2$.  The numerical results are averaged  over $50$ random signal realizations. We can notice from Fig.  \ref{fig:conf_isufi} that, while our filter is insensitive to the harmonic component, the filtering strategy proposed in \cite{Yang2021} is affected by the harmonic component.
\begin{figure}[t]
\centering
\includegraphics[width=8.5cm,height=5cm]{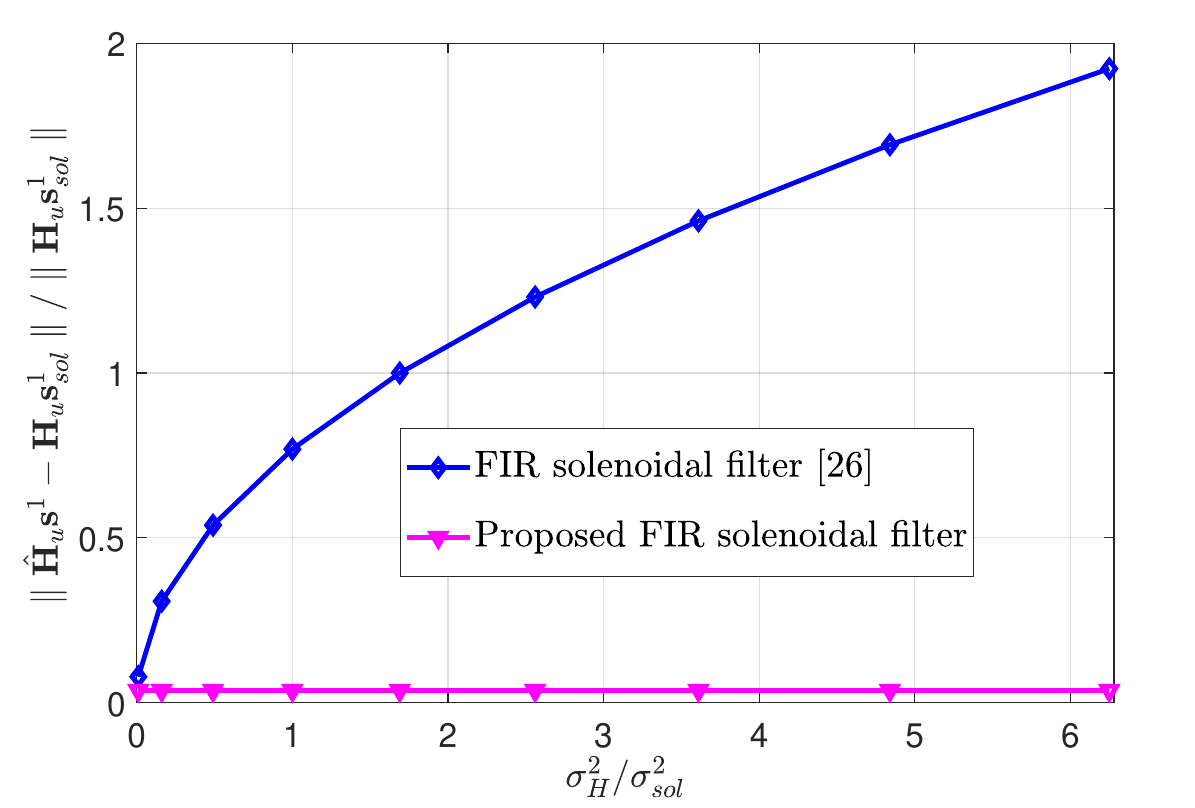}
\caption{NMSE of the filtered solenoidal signal vs the ratio $\frac{\sigma_H^2}{\sigma_{sol}^2}$.}
\label{fig:conf_isufi}
\end{figure}
    \begin{figure}[t]
\centering
\includegraphics[width=8.5cm,height=5.3cm]{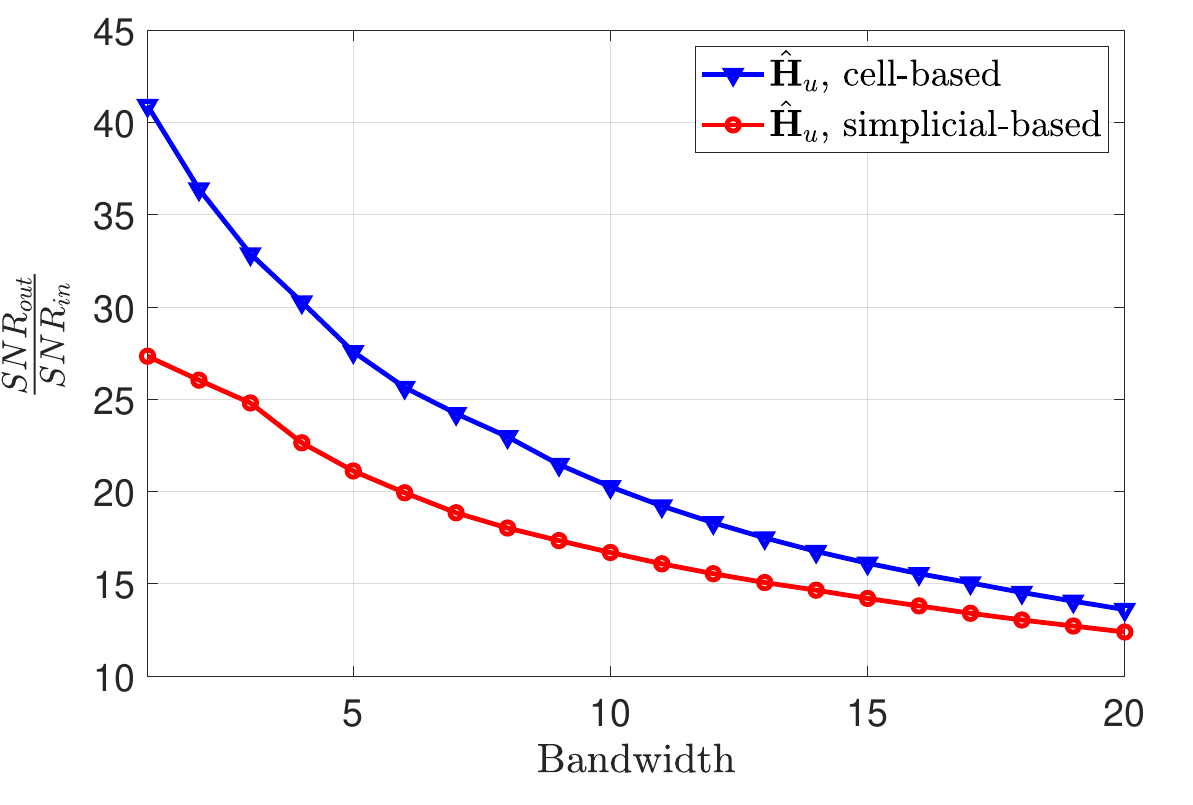}
\caption{SNR improvement factor vs the signal bandwidth.}
\label{fig:SNR_impr}
\end{figure}

To investigate the benefits of using cell-based against simplicial-based filters over  real data, 
we considered the transportation network and the associated real traffic from the  city of Anaheim \cite{Stab_data}, in the year 1992. This network is composed of $N=416$ nodes, $E=914$ links and $208$ polygons.
The flow $f_{ij}$ over each link $(i,j)$ is bidirectional so we define the flow signal over edge $(i,j)$ as  $s_{ij}=(f_{ij}+f_{j i})/2$.
Then, we built the Laplacian matrices of the simplicial complex and of the cell complexes associated to the graph 
and we computed the associated eigenvector matrices  $\mU_{sc}$ and $\mU_1$, respectively.
The two complexes have been built by filling all possible  $T=54$ triangles in the SC case and by inferring the cell complex with $198$ polygons in the CC case. Then, we projected the observed edge flow $\bs^1$ onto the two sets of eigenvectors $\mU_{sc}$ and $\mU_1$, giving rise to the two vectors
$\hat{\bs}^1_{sc}=\mU_{sc}^T \bs^1$ and $\hat{\bs}^1_{cc}=\mU_{1}^T \bs^1$.
We then build two diagonal selection matrices $\mD_{sc}$ and $\mD_1$ with  entries equal to $1$ in the positions associated with the largest $B$ frequency coefficients, respectively of $\hat{\bs}^1_{sc}$ and $\hat{\bs}^1_{cc}$.
We added a random zero-mean Gaussian noise with unit variance  to the observed flow   and then we filtered the resulting signal using the two filters $\mH_{sc}=\mU_{sc}\mD_{sc} \mU_{sc}^T$ and $\mH_{1}=\mU_{1}\mD_{1} \mU_{1}^T$. In Fig. \ref{fig:SNR_impr}
we report the SNR improvement factor  $\text{SNR}_{out}/\text{SNR}_{in}$, i.e. the ratio between the output and the input Signal-to-Noise ratio, versus the signal bandwidth $B$.
We can observe that the use of CC-based filter yields a substantial performance gain compared  to the SC-based filter, due to its better capability of filling a surface using polygons of arbitrary dimension.

\section{Conclusion}
In this paper, we proposed an algebraic framework for analyzing signals represented over cell complexes by showing the significant advantages of using cell spaces compared to simplicial complexes. \textcolor{black}{We introduced planar CC whose cells can contain holes, as a suitable framework to handle images or (approximately) piecewise signals.}
We provided efficient methods to infer a cell complex from flow signals and investigated alternative strategies to filter their harmonic, solenoidal, and irrotational components. We derived the theoretical conditions guaranteeing a performance gain associated with the independent design of solenoidal and irrotational FIR filters, as opposed to a common FIR filter design. 
\textcolor{black}{
Further investigations should address the incorporation of the proposed filtering and inference capabilities within topological deep learning frameworks, such as graph neural networks, and a more in-depth analysis of the image coding capabilities of the new (hollow) cell complexes.}


\appendices{}

\section{Proof of Theorem I}
\noindent We start by introducing all the notation.\\
\textcolor{black}{\textbf{Incidence matrices}.
Let us consider a second order planar cell complex $\mathcal{X}$ obtained by partitioning a finite bounded space.
The $1$-skeleton $\mathcal{G}=\{\mathcal{V},\mathcal{E}\}$ of $\mathcal{X}$ is composed of the union of  $n_c \geq 1$  disconnected graphs
$\mathcal{G}_i=\{\mathcal{V}_i, \mathcal{E}_i\}$, with $N_i=|\mathcal{V}_i|$ and $E_i=|\mathcal{E}_i| $.  Then,  we denote by $N=|\mathcal{V}|= \sum_{i=1}^{n_c} N_i$ and $E=|\mathcal{E}|= \sum_{i=1}^{n_c} E_i$, respectively,  the number of nodes and edges of $\mathcal{G}$.
We associate a 
subcomplex $\mathcal{C}_i=\{\mathcal{V}_i,\mathcal{E}_i,\mathcal{F}_i\}$ to each graph component $\mathcal{G}_i$, where $\mathcal{F}_i$ is the set of the $2$-order cells of $\mathcal{C}_i$, with $|\mathcal{F}_i|=F_i$. Let us also denote by $\mathcal{Q}_i$ and $\mathcal{P}_i$, respectively, the set of $2$-order cells with and without holes, so that $\mathcal{F}_i$ can be written as $\mathcal{F}_i=\mathcal{Q}_i\cup \mathcal{P}_i$.
Then, we define  $\mB_1^{(i)}$ and $\mB_2^{(i)}$ as  the node-edge  and the edge-polygon  incidence matrices of the sub-complex $\mathcal{C}_i$ associated with the single graph $\mathcal{G}_i$. The cell complex $\mathcal{X}=\{\mathcal{V}, \mathcal{E}, \mathcal{F}\}$ admits then an algebraic representation through its  incidence matrices $\mB_1$ and $\mB_2$. More specifically, since the $1$-skeleton of $\mathcal{X}$ is composed of $n_c$ disconnected graphs,  its node-edge incidence matrix $\mB_1$ has a  block diagonal structure with diagonal blocks being the  incidence matrices $\mB_1^{(i)}$ of each graph, i.e. 
\beq \label{eq:B1_bl}
\mB_1=\text{blkdiag}(\mB_1^{(1)},\ldots,\mB_1^{(n_c)}).
\eeq
Let us now define the structure of the edge-polygon  incidence matrix $\mB_2$ associate with $\mathcal{X}$. 
The orientation of the $2$-order cells is chosen such that two neighbor $2$-order cells induce opposite orientations on the shared $1$-cell boundary elements. The set $\mathcal{F}$ of the $2$-order cells  in $\mathcal{X}$ can be written as the union of the two sets $\mathcal{P}=\cup_{i=1}^{n_c}\mathcal{P}_i$   and $\mathcal{Q}=\cup_{i=1}^{n_c}\mathcal{Q}_i$, containing, respectively, the cells without and with holes in their interiors. Then, we denote by $|\mathcal{P}|=P$ and 
$|\mathcal{Q}|=Q$, the number of cells without and with holes. 
Every cell $c_k^{2}$ without holes,  $ \forall k\in \mathcal{P}$, 
is identified by its outer boundary; each hollow cell $c_k^{2}$,  with $n_k\geq 1$ holes,  for $k \in \mathcal{Q}$,  has a boundary $\partial c_k^{2}$ given by the union of one outer boundary and one or multiple interior boundaries, one for each hole.  Note that, according to the adopted orientation of $2$-order cells, such that any two neighbor cells have an opposite sign over each boundary edge,
the edges associated with the outer boundary of each subcomplex $\mathcal{C}_i$ can be represented as a  vector $\mathbf{q}^{(i)}=\mB_2^{(i)} \mathbf{1}_{F_i}$ obtained by adding all the columns of $\mB_2^{(i)}$. To simplify our notation, we associate with each $2$-order cell  $c_k^{2}$ a set of indices $\mathcal{H}_k=\{\mathcal{I}_k, \ell_k, h_k\}$ where: i) $\mathcal{I}_k \subset \{1,\ldots,n_c\}$ is the set of indices of the connected components associated with the inner boundaries of $c_k^{2}$; ii)  $\ell_k$ is the index of the connected component to which the outer boundary of $c_k^{2}$ belongs, and $h_k$ is the index in $\mathcal{Q}_{\ell_k}$ associated with the cell $c_k^{2}$.
Then, the edges/$2$-cells incidence matrix $\mB_2$ of dimension $E \times (P+Q)$ has the following structure
\beq \label{eq:B_2_planar}
\mB_2=\left[\begin{array}{c|c}
\begin{matrix}\mB_{2, \mathcal{P}_1}^{(1)} & \mathbf{O} & \cdots &\mathbf{O}\\
\mathbf{O} &\mB_{2, \mathcal{P}_2}^{(2)}&\cdots&\mathbf{O}\\
\vdots &\vdots&\vdots&\vdots \\
\mathbf{O} & \cdots & \cdots & \mB_{2, \mathcal{P}_{n_c}}^{(n_c)}\end{matrix}& 
\begin{matrix}\bar{\mathbf{b}}_{2,1}^{(1)}& \cdots&\bar{\mathbf{b}}_{2,Q}^{(1)}\\ \bar{\mathbf{b}}_{2,1}^{(2)} & \cdots &\bar{\mathbf{b}}_{2,Q}^{(2)}\\
\vdots &\vdots&\vdots \\ \bar{\mathbf{b}}_{2,1}^{(n_c)}& \cdots &\bar{\mathbf{b}}_{2,Q}^{(n_c)} \end{matrix}\end{array}
\right].
\eeq
The first $P$ columns define a block diagonal matrix where each block $\mB_{2, \mathcal{P}_i}^{(i)}$  has dimension $E_i \times P_i$ and is composed of the columns of  $\mB_{2}^{(i)}$ with indices in $\mathcal{P}_i$, i.e.,
corresponding to the $2$-order cells without holes in the component $\mathcal{G}_i$. The last $Q$ columns of $\mB_2$ determine   the outer and inner boundaries of the $Q$ cells with holes, by defining  the vectors $\bar{\mb}_{2,k}^{(j)}$ for $k=1,\ldots,Q$ as:
\beq \label{eq:bar_bi}    \bar{\mb}_{2,k}^{(j)}=\left\{ \begin{array}{lll}
    \mathbf{0}_{E_j} & \text{if}  \, j \notin \mathcal{I}_k, j\neq \ell_k, \medskip \\
 -{\mathbf{q}}^{(j)} & \forall j \in \mathcal{I}_k,
\medskip \\    
    \mB_2^{(j)}(:,h_k) & \text{if} \, j=\ell_k.
    \end{array}\right.
    \eeq
\textbf{Hodge Laplacian decomposition.}  After defining the incidence matrices, we can now introduce the Laplacian matrices.
The zero-order Laplacian $\mL_0=\mB_1 \mB_1^T$ can be  written as a block diagonal matrix with a number of blocks equal to the number of connected components of the graph.
Since each block $\mB_1^{(i)}\mB_1^{(i),T}$ is the $0$-order Laplacian associated with a connected graph, its rank is equal to $N_i-1$. Then,
the rank of $\mL_0$ is equal to $N-n_c$, and, as expected,  the dimension of $\text{ker}(\mL_0)$ is equal to the number of connected components of the graph $\mathcal{G}$.\\
  Let us now consider the lower Laplacian $\mL_{1,d}=\mB_1^T \mB_1$. From (\ref{eq:B1_bl}),  we easily get  $\mL_{1,d}=\mB_1^{T}\mB_1=\text{bkldiag}(\mB_1^{(1), T}\mB_1^{(1)}, \ldots,\mB_1^{(n_c), T}\mB_1^{(n_c)})$, then also $\mL_{1,d}$ is a block diagonal matrix. Furthermore, since it holds $\text{rank}(\mB_1^T \mB_1)=\text{rank}(\mB_1 \mB_1^T)$, the dimension of the (irrotational) space spanned by the eigenvectors of $\mL_{1,d}$ is  equal $N-n_c$. 
Denoting with $\mL_{1,u}=\mB_2 \mB_2^T$ the upper Laplacian matrix, in  the following we will show that the structure of a  planar cell complex whose cells may contain holes is still described by the first order Laplacian $\mL_1=\mL_{1,d}+\mL_{1,u}$, which induces the Hodge decomposition of the edge space.
Note that,  using the orthogonality condition $\mB_1^{(i)}\mB_2^{(i)}=0$,  it can be easily shown that 
 $\mB_1 \mB_2=\text{bkldiag}(\mB_1^{(1)},\mB_1^{(2)})\mB_2=\mathbf{0}$.
}
\textcolor{black}{\textbf{Proof of Theorem 1.} 
To prove point a) note that the vectors $\bar{\mathbf{b}}_{2,k}^{(i)}$ in (\ref{eq:bar_bi}) are either columns of the matrix ${\mathbf{B}}_{2}^{(i)}$, or their sum, or all zero vectors.
Then, 
 it  holds $\mB_1^{(i)}\bar{\mathbf{b}}_{2,k}^{(i)}=\mathbf{0}$, $\forall i,k$.   
Hence, using (\ref{eq:B_2_planar}) and the orthogonality conditions $\mB_1^{(i)}\mB_{2, \mathcal{P}_i}^{(i)}=\mathbf{0}$,
we get 
\beq \label{eq:orth_condn}
\mB_1 \mB_2=\text{blkdiag}(\mB_1^{(1)}, \ldots,\mB_1^{(n_c)} ) \mB_2=\mathbf{0},
\eeq
and this concludes the proof of  point a).
To prove point b) we can  apply Theorem 5.2 in \cite{Lim} stating that given two matrices $\mA\in \mathbb{R}^{m\times n}$ and $\mB\in \mathbb{R}^{n\times p}$ with $\mA \mB=\mathbf{0}$,
it holds $\mathbb{R}^n = \text{img}(\mA^T) \oplus \text{ker}(\mA^T \mA +\mB \mB^T) \oplus \text{img}(\mB)$. Then, from (\ref{eq:orth_condn}), we get the Hodge decomposition of the edge space as
\beq
\mathbb{R}^E = \text{img}(\mB_1^T) \oplus \text{ker}(\mL_1) \oplus \text{img}(\mB_2)
\eeq
where $\mL_1=\mB_1^T \mB_1 +\mB_2 \mB_2^T$ is the first order Laplacian matrix.
The subspace $\text{img}(\mB_1^T)$ and $\text{img}(\mB_2)$  are called, respectively, the irrotational and solenoidal subspaces, while $\text{ker}(\mL_1)$ is the harmonic subspace. From Theorem 5.2 in \cite{Lim}, it also holds that  $\text{ker}(\mL_1)=\text{ker}(\mB_1^T) \cap \text{ker}(\mB_2)$ and there is a natural isomorphism between $\text{ker}(\mL_1)$ and  the (co)homology groups
 $\text{ker}(\mB_1)/\text{ker}(\mB_2)$ identifying topological invariants (see  Theorem 5.3 \cite{Lim}). Then,  we  can now prove that   the kernel of the Hodge Laplacians associated with  a cell complex having cells not homeomorphic to $2D$-balls provides information on the graph connectivity and on the number of holes in the complex.
 Considering the 0-order Laplacian $\mL_0=\text{blkdiag}(\mB_1^{(1)} \mB_1^{(1),T}, \ldots,\mB_1^{(n_c)} \mB_1^{(n_c),T})$, we have  $\text{rank}(\mL_0)=N-n_c$. Thus, the dimension of $\text{ker}(\mL_0)$ represents the number of connected 
components of the graph. Let us now consider the lower Laplacian matrix $\mL_{1,d}=\mB_1^T \mB_1=\text{blkdiag}(\mB_1^{(1), T} \mB_1^{(1)}, \ldots,\mB_1^{(n_c),T} \mB_1^{(n_c)})$. Since  each  matrix $\mB_1^{(i), T} \mB_1^{(i)}$ has rank $N_i-1$,  the dimension of the irrotational subspace spanned  by the eigenvectors of $\mL_{1,d}$ is equal to $\sum_{i=1}^{n_c}N_i-1=N-n_c$.\\
Let us now consider the dimension of the solenoidal subspace spanned by the eigenvectors of the upper Laplacian matrix $\mL_{1,u}=\mB_2 \mB_2^T$.
From the Hodge decomposition it follows that  $E=\text{rank}(\mL_{1,d})+\text{rank}(\mL_{1,u})+\text{dim}(\text{ker}(\mL_1))$.  Therefore, when $\text{dim}(\text{ker}(\mL_1))=0$, the maximum rank $r_{\text{max}}$ of the upper Laplacian matrix $\mL_{1,u}$, when all the $2$-order cells are filled, is 
\beq \label{eq:rank_L1u}
r_{\text{max}}=E-N+n_c.
\eeq
On the other hand, the incidence matrix $\mB_2$ is an $E \times F$-dimensional matrix with $F$ the number of possible $2$-order cells in the complex. 
From the Euler's Theorem for planar graphs with $n_c$ connected components \cite{gallier2011discrete}, we have that the number of possible faces, excluding the unbounded external face,  is 
\beq 
\label{eq:Euler}F=E-N+n_c. \eeq
Then, when $\mB_2$ contains all the possible $F$ independent columns,    its rank is maximum, and using (\ref{eq:rank_L1u}) and (\ref{eq:Euler}),  we get $r_{\text{max}}=E-N+n_c=F$. Then, from the Hodge decomposition, it holds that $\text{dim}(\text{ker}(\mL_1))=0$ when all  $2$-order cells are filled. If the complex contains $n_h$ holes, the corresponding columns of $\mB_2$ are missing, and as a result, the rank is reduced by $n_h$. 
Therefore, 
the dimension of the harmonic subspace, called also first Betti number, represents for a   complex partitioning a closed planar space the number of $1$-dimensional holes in the complex, how it happens for cell complexes with connected $1$-skeletons \cite{friedman1996computing}, \cite{grady2010}.}
\vspace{0.2cm}
\section{Optimal  solutions of problems (32) and (33)}
The two least squares problems (32), (33) admit always  at least one solution with minimum norm given by the  pseudoinverse in (34). Furthermore, if the two  matrices $\boldsymbol{\Phi}_{s}$ and $\boldsymbol{\Phi}_{I}$ of size $n_u \times L$ and $n_I \times L$, respectively, are fat so that $L>n_u$ and $L>n_d$, then, due to their structure in (28), their ranks are equal to $n_u$ and $n_d$.  This implies that $\text{rank}([\boldsymbol{\Phi}_{s}, \mathbf{h}^s])=\text{rank}(\boldsymbol{\Phi}_{s})$ and
$\text{rank}([\boldsymbol{\Phi}_{I}, \mathbf{h}^I])=\text{rank}(\boldsymbol{\Phi}_{I})$ so that the two systems $\boldsymbol{\Phi}_{s} \ba^s=\mathbf{h}^s$, $\boldsymbol{\Phi}_{I} \ba^I=\mathbf{h}^I$ are consistent, i.e. $f_s(\ba^{s \star})=f_I(\ba^{I \star})=0$.
If $L=n_u$ and/or $L=n_d$ then  the solution of  (32) and/or (33) are unique and given by $\ba^{s \star}=(\boldsymbol{\Phi}_{s}^T \boldsymbol{\Phi}_{s})^{-1} \boldsymbol{\Phi}_{s}^T \mathbf{h}^s$ and $\ba^{I \star}=(\boldsymbol{\Phi}_{I}^T \boldsymbol{\Phi}_{I})^{-1} \boldsymbol{\Phi}_{I}^T \mathbf{h}^I$.

\vspace{0.2cm}
\section{Proof of Proposition 2}
\label{sec:Prop_2}
Le us consider the convex quadratic functions $f(\mathbf{a})$, $f_s(\mathbf{a}^s)$  and  $f_I(\mathbf{a}^I)$  in (30), (33) and (32).
Using (29), we have $f(\mathbf{a})=f_I(\mathbf{a})+f_s(\mathbf{a})$, then problem in (30) can be also written as
\beq \label{eq:comb}
\underset{\ba \in \mathbb{R}^{L}}{\min} \qquad f(\mathbf{a})=f_I(\mathbf{a})+f_s(\mathbf{a}).
\eeq
Denote with $f(\mathbf{a}^\star)$  the minimum value of the objective function of  the convex problem (\ref{eq:comb}) at the optimal point $\mathbf{a}^{\star}$.
If $\mathbf{a}^{I \, \star}$ and $\mathbf{a}^{s \, \star}$ are two optimal points of the convex problems  (33) and (32),  then it holds
\beq
f_I(\mathbf{a}^{I \, \star})\leq f_I(\mathbf{a}), \quad \text{and} \quad f_s(\mathbf{a}^{s \, \star})\leq f_s(\mathbf{a}), \quad \forall \mathbf{a} \in \mathbb{R}^L
\eeq
so that
\beq
f_I(\mathbf{a}^{I \, \star})+f_s(\mathbf{a}^{s \, \star})\leq f_I(\mathbf{a})+f_s(\mathbf{a})=f(\mathbf{a}), \quad \forall \mathbf{a} \in \mathbb{R}^L.
\eeq
 From  (\ref{eq:comb}), this implies that
\beq \label{eq:ineq_disj}
f_I(\mathbf{a}^{I \, \star})+f_s(\mathbf{a}^{s \, \star})\leq f(\mathbf{a}^{\star})
\eeq
so that the error in the filter estimation using separate FIR filters is upper bounded by the estimation error obtained with a single  filter.
Let us now prove point a) of the proposition, i.e. assuming the system $\boldsymbol{\Phi}_{t}  \ba =\mathbf{h}$ to be inconsistent, we provide sufficient conditions under which  the inequality in (\ref{eq:ineq_disj}) is strict, i.e.
\beq
f_I(\mathbf{a}^{I \, \star})+f_s(\mathbf{a}^{s \, \star}) < f(\mathbf{a}^{\star}).
\eeq
First, we may observe that the error function $f(\ba)$ can also be written as
\beq \label{eq:l2}
\begin{split}
&f(\ba)=\parallel \mathbf{H}-\ds \sum_{k=1}^{K} a_k\, \mL_1^k  \parallel^2_F  = \parallel \text{vec}(\mathbf{H})-\ds \sum_{k=1}^{K} a_k\, \text{vec}(\mL_1^k)   \parallel^2_F\\& = \parallel  \text{vec}(\mathbf{H})-\ds \sum_{k=1}^{K} a_k\, [\text{vec}(\mL_{1,d}^k)+\text{vec}(\mL_{1,u}^k)]  \parallel^2_F\\
& = \parallel  \text{vec}(\mathbf{H})- (\boldsymbol{\Phi}_{d}+\boldsymbol{\Phi}_{u}) \mathbf{a} \parallel^2_F
\end{split}
\eeq
where the last equality follows from (35).

Then, an optimal point $\ba^{\star}$ minimizing the least square function in (\ref{eq:l2}) is given by
\beq \label{eq:opt_point}
\mathbf{a}^{\star}=(\boldsymbol{\Phi}_{d}+\boldsymbol{\Phi}_{u})^{\dag}\text{vec}(\mathbf{H}).
\eeq
Similarly, since it also holds
\beq \label{eq:l4}
 f_s(\ba^s)=\parallel  \text{vec}(\mathbf{H}_{u})- \boldsymbol{\Phi}_{u} \mathbf{a}^s \parallel^2_F, \;  f_I(\ba^I)=\parallel  \text{vec}(\mathbf{H}_{d})- \boldsymbol{\Phi}_{d} \mathbf{a}^I \parallel^2_F
\eeq
we can easily prove that the optimal solutions of the optimization problems in (32) and (33) are
\beq \label{eq:opt_s}
\mathbf{a}^{s\, \star}=\boldsymbol{\Phi}_{u}^{\dag}\text{vec}(\mathbf{H}_{u}), \qquad
\mathbf{a}^{I\, \star}=\boldsymbol{\Phi}_{d}^{\dag}\text{vec}(\mathbf{H}_{d}).
\eeq
Then, replacing in (\ref{eq:ineq_disj}) the optimal solutions in  (\ref{eq:opt_point}) and (\ref{eq:opt_s}), we obtain
\beq \label{eq:ineq_norm}
\begin{split}
& \parallel  [\mI - \boldsymbol{\Phi}_{d}  \boldsymbol{\Phi}_{d}^{\dag}]\text{vec}(\mathbf{H}_{d}) \parallel^2_F+ \parallel [\mI - \boldsymbol{\Phi}_{u}  \boldsymbol{\Phi}_{u}^{\dag}] \text{vec}(\mathbf{H}_{u}) \parallel^2_F \leq \\
& \parallel [\mI - (\boldsymbol{\Phi}_{d}+\boldsymbol{\Phi}_{u})(\boldsymbol{\Phi}_{d}+\boldsymbol{\Phi}_{u})^{\dag}] \text{vec}(\mathbf{H}) \parallel^2_F=\\ & \parallel [\mI - (\boldsymbol{\Phi}_{d}+\boldsymbol{\Phi}_{u})(\boldsymbol{\Phi}_{d}+\boldsymbol{\Phi}_{u})^{\dag}] (\text{vec}(\mathbf{H}_{d})+ \text{vec}(\mathbf{H}_{u}))\parallel^2_F.
\end{split}
\eeq
In order to show under which conditions  this last inequality holds strict, let us first derive some properties of the  subspace spanned by the columns of the matrix $\boldsymbol{\Phi}_{d}+\boldsymbol{\Phi}_{u}$.
Using the identity $\text{vec}(\mA \mB \mC)=(\mC^T \otimes \mA) \text{vec}(\mB)$ and from (26), we can write
\beq \label{eq:vec_H}
\begin{split}
&\text{vec}(\mH)=\text{vec}(\mH_{u})+\text{vec}(\mH_{d})\\
& = (\mU_{u} \otimes \mU_{u})\text{vec}(h^s(\boldsymbol{\Lambda}_{u}))+ (\mU_{d} \otimes \mU_{d})\text{vec}(h^I(\boldsymbol{\Lambda}_{d})).
\end{split}
\eeq
Hence  (35) can be also written as
\beq
\begin{split}
& \boldsymbol{\Phi}_{d}=(\mU_{d} \otimes \mU_{d}) [\text{vec}(\boldsymbol{\Lambda}_{d}),\ldots, \text{vec}(\boldsymbol{\Lambda}_{d}^K)]\\
& \boldsymbol{\Phi}_{u}=(\mU_{u} \otimes \mU_{u}) [\text{vec}(\boldsymbol{\Lambda}_{u}),\ldots, \text{vec}(\boldsymbol{\Lambda}_{u}^K)]\end{split}
\eeq
so that it holds:
\beq \label{eq:ort_cond}
\boldsymbol{\Phi}_{d}^T \boldsymbol{\Phi}_{u}=\mathbf{0}, \quad \text{and} \quad
\boldsymbol{\Phi}_{u}^T \boldsymbol{\Phi}_{d}=\mathbf{0}.
\eeq
Then, defining the  subspaces of $\mathbb{R}^{E^2}$, $\mathcal{N}=\mathcal{R}(\boldsymbol{\Phi}_{d})$, $\mathcal{Q}=\mathcal{R}(\boldsymbol{\Phi}_{u})$ and $\mathcal{M}=\mathcal{R}(\boldsymbol{\Phi}_{d}+\boldsymbol{\Phi}_{u})$, let us consider the  projections $\mP_{(\mathcal{N}+\mathcal{Q})^{\perp}}$, $\mP_{\mathcal{M}^{\perp}}$ onto the orthogonal subspaces $(\mathcal{N}+\mathcal{Q})^{\perp}$ and $\mathcal{M}^{\perp}$, respectively, given by
\beq \label{eq:M_ort}
\mP_{\mathcal{M}^{\perp}}=\mI-(\boldsymbol{\Phi}_{d}+\boldsymbol{\Phi}_{u})(\boldsymbol{\Phi}_{d}+\boldsymbol{\Phi}_{u})^{\dag}
\eeq
and
\beq \label{eq:sum_sub}
\mP_{(\mathcal{N}+\mathcal{Q})^{\perp}}=\mI-\mP_{\mathcal{N}+\mathcal{Q}}.
\eeq
Since (see \cite{bernstein2005}[pp. 234]), $\mA^{\dag} \mB=\mathbf{0}$ if and only if $\mA^T \mB=\mathbf{0}$,  using  (\ref{eq:ort_cond}), it follows that \beq \label{eq:phi_ort}
\begin{split} & \boldsymbol{\Phi}_{d}^{\dag} \boldsymbol{\Phi}_{u}=\mathbf{0}, \quad \boldsymbol{\Phi}_{u}^{\dag} \boldsymbol{\Phi}_{d}=\mathbf{0},
\end{split}
\eeq then, defining the projections $\mP_{\mathcal{N}}=\boldsymbol{\Phi}_{d} \boldsymbol{\Phi}_{d}^{\dag}$, and $\mP_{\mathcal{Q}}=\boldsymbol{\Phi}_{u}\boldsymbol{\Phi}_{u}^{\dag}$, we have
\beq
\mP_{\mathcal{N}} \mP_{\mathcal{Q}}= \boldsymbol{\Phi}_{d} \boldsymbol{\Phi}_{d}^{\dag} \boldsymbol{\Phi}_{u}\boldsymbol{\Phi}_{u}^{\dag}=\mathbf{0}
\eeq
or, equivalently, $\mP_{\mathcal{N}} \cap \mP_{\mathcal{Q}}=\emptyset$.
As proved in \cite{piziak1999}[pp. 72], if $\mP_{\mathcal{N}} \cap \mP_{\mathcal{Q}}=\emptyset$, then it holds
\beq \label{eq:sumP}
\mP_{\mathcal{N}+\mathcal{Q}}=\mP_{\mathcal{N}}+\mP_{\mathcal{Q}},
\eeq
so (\ref{eq:sum_sub}) reduces to
\beq \label{eq:sum1}
\mP_{(\mathcal{N}+\mathcal{Q})^{\perp}}=\mI-\mP_{\mathcal{N}}-\mP_{\mathcal{Q}}.
\eeq
To prove point a) of the Proposition, we need first to show that:
\begin{itemize}
    \item[1)] The following   strict inclusion among subspaces holds:
    \beq \label{eq:range_inclus1}
\mathcal{R}(\boldsymbol{\Phi}_{d}+\boldsymbol{\Phi}_{u}) \subset \mathcal{R}(\boldsymbol{\Phi}_{d})+\mathcal{R}(\boldsymbol{\Phi}_{u});
\eeq
\item[2)] For any $\bx \in \mathbb{R}^{E^2}$  we have
\beq \label{eq:dis3}
\parallel \mathbf{P}_{(\mathcal{N}+\mathcal{Q})^{\perp}} \bx \parallel^2 \leq \parallel \mathbf{P}_{\mathcal{M}^{\perp}} \bx \parallel^2.
\eeq
\end{itemize}
Let us prove the first point. Given the  two matrices $\boldsymbol{\Phi}_{d}$ and $\boldsymbol{\Phi}_{u}$, it always holds \cite{Meyer}[pag. 205]
\beq \label{eq:rank_d}
\text{rank}(\boldsymbol{\Phi}_{d}+\boldsymbol{\Phi}_{u}) \leq \text{rank}(\boldsymbol{\Phi}_{d})+ \text{rank}(\boldsymbol{\Phi}_{u}).
\eeq
Let us now show that the inequality in (\ref{eq:rank_d}) holds strict.
Note that $\boldsymbol{\Phi}_{d}$ and $\boldsymbol{\Phi}_{u}$ are full column rank  matrices with $\text{rank}(\boldsymbol{\Phi}_{u})=\text{rank}(\boldsymbol{\Phi}_{d})=K$, and $\text{max}(\text{rank}(\boldsymbol{\Phi}_{d}+\boldsymbol{\Phi}_{u}))=K$, then 
 it  follows that the inequality in (\ref{eq:rank_d}) holds strictly, i.e.
 \beq \label{eq:rank_ineq}
 \text{rank}(\boldsymbol{\Phi}_{d}+\boldsymbol{\Phi}_{u}) < \text{rank}(\boldsymbol{\Phi}_{d})+ \text{rank}(\boldsymbol{\Phi}_{u}).
 \eeq
Furthermore, 
it holds
 \cite{Meyer}[pag. 206]
\beq \label{eq:range_inclus}
\mathcal{R}(\boldsymbol{\Phi}_{d}+\boldsymbol{\Phi}_{u}) \subseteq \mathcal{R}(\boldsymbol{\Phi}_{d})+\mathcal{R}(\boldsymbol{\Phi}_{u}).
\eeq
From this last equation, it follows that \cite{Meyer}[pag. 198]
\beq \label{eq:dis2}
\text{dim}(\mathcal{R}(\boldsymbol{\Phi}_{d}+\boldsymbol{\Phi}_{u})) \leq \text{dim}(\mathcal{R}(\boldsymbol{\Phi}_{d})+\mathcal{R}(\boldsymbol{\Phi}_{u})).
\eeq
Then, using the equality
\beq \label{eq:dis1}
\text{rank}(\boldsymbol{\Phi}_{d}+\boldsymbol{\Phi}_{u})= \text{dim}(\mathcal{R}(\boldsymbol{\Phi}_{d}+\boldsymbol{\Phi}_{u})),
\eeq
we can write
\beq \label{eq: disf}
\begin{split}
\text{rank}(\boldsymbol{\Phi}_{d}+\boldsymbol{\Phi}_{u})&\underset{(a)}{=} \text{dim}(\mathcal{R}(\boldsymbol{\Phi}_{d}+\boldsymbol{\Phi}_{u})) \\
& \underset{(b)}{\leq}  \text{dim}(\mathcal{R}(\boldsymbol{\Phi}_{d})+\mathcal{R}(\boldsymbol{\Phi}_{u}))\\
&\underset{(c)}{=} \text{dim}(\mathcal{R}(\boldsymbol{\Phi}_{d}))+\text{dim}(\mathcal{R}(\boldsymbol{\Phi}_{u})- \\
& \quad \text{dim}(\mathcal{R}(\boldsymbol{\Phi}_{d}) \cap \mathcal{R}(\boldsymbol{\Phi}_{u}))
\end{split}
\eeq
where equality (a) follows from (\ref{eq:dis1}), (b)  from   (\ref{eq:dis2}) and (c) from the following  equality \cite{Meyer}[pag. 205]
\beq
\text{dim}(\mathcal{A}+\mathcal{B})=
\text{dim}(\mathcal{A})+\text{dim}(\mathcal{B})-\text{dim}(\mathcal{A}\cap \mathcal{B})
\eeq
where $\mathcal{A}$ and $\mathcal{B}$ are subspaces of a vector space $\mathcal{V}$.
Let us now prove that
\beq \label{eq:orth_cond}
\mathcal{R}(\boldsymbol{\Phi}_{d}) \cap \mathcal{R}(\boldsymbol{\Phi}_{u})=\emptyset.
\eeq
Using (\ref{eq:ort_cond})  $\forall \bb_u \in \mathcal{R}(\boldsymbol{\Phi}_{u})$ with $\bb_u=\boldsymbol{\Phi}_{u} \bx_{u}$, and
$\forall \bb_d \in \mathcal{R}(\boldsymbol{\Phi}_{d})$ with $\bb_d=\boldsymbol{\Phi}_{d} \bx_{d}$, we have  $\bb_d^T \bb_u=0$ and, this proves equation (\ref{eq:orth_cond}).
Hence, (\ref{eq: disf}) reduces to
\beq \label{eq: disf1}
\begin{split}
\text{rank}(\boldsymbol{\Phi}_{d}+\boldsymbol{\Phi}_{u})&\underset{(a)}{=} \text{dim}(\mathcal{R}(\boldsymbol{\Phi}_{d}+\boldsymbol{\Phi}_{u})) \\
& \underset{(b)}{\leq}  \text{dim}(\mathcal{R}(\boldsymbol{\Phi}_{d})+\mathcal{R}(\boldsymbol{\Phi}_{u}))\\
& {=} \, \text{dim}(\mathcal{R}(\boldsymbol{\Phi}_{d}))+\text{dim}(\mathcal{R}(\boldsymbol{\Phi}_{u})\\ & {=} \, \text{rank}(\boldsymbol{\Phi}_{d})+
\text{rank}(\boldsymbol{\Phi}_{u})
\end{split}
\eeq
and from (\ref{eq:rank_ineq}), it follows that the inequality in (b) holds strict, i.e.
\beq
\text{dim}(\mathcal{R}(\boldsymbol{\Phi}_{d}+\boldsymbol{\Phi}_{u})) < \text{dim}(\mathcal{R}(\boldsymbol{\Phi}_{d}))+\text{dim}(\mathcal{R}(\boldsymbol{\Phi}_{u})).
\eeq
This last inequality together with (\ref{eq:range_inclus}) implies    equation (\ref{eq:range_inclus1}), i.e.
\beq \label{eq:range_inclus2}
\mathcal{R}(\boldsymbol{\Phi}_{d}+\boldsymbol{\Phi}_{u}) \subset \mathcal{R}(\boldsymbol{\Phi}_{d})+\mathcal{R}(\boldsymbol{\Phi}_{u}).
\eeq
Let us now prove  the inequality in (\ref{eq:dis3}).
It is known that given two subspaces $\mathcal{A}$,
$\mathcal{B}$, then $\mathcal{A}\subset \mathcal{B}$ if and only if $\mathcal{B}^{\perp} \subset \mathcal{A}^{\perp}$ \cite{bernstein2005}[pag. 48].
Therefore,  from (\ref{eq:range_inclus2}) we get
\beq \label{eq:range_inclus3}
 (\mathcal{R}(\boldsymbol{\Phi}_{d})+\mathcal{R}(\boldsymbol{\Phi}_{u}))^{\perp} \subset \mathcal{R}^{\perp}(\boldsymbol{\Phi}_{d}+\boldsymbol{\Phi}_{u})
\eeq
or  according to our notation
\beq \label{eq:sub_incl}
 (\mathcal{N}+\mathcal{Q})^{\perp}  \subset \mathcal{M}^{\perp}.
\eeq
Therefore, from (\ref{eq:sub_incl}),  for any $\bx \in \mathbb{R}^{E^2}$  it follows the inequality in (\ref{eq:dis3}), i.e.
\beq \label{eq:dis4}
\parallel \mathbf{P}_{(\mathcal{N}+\mathcal{Q})^{\perp}} \bx \parallel^2 \leq \parallel \mathbf{P}_{\mathcal{M}^{\perp}} \bx \parallel^2.
\eeq
Let us now derive some sufficient conditions under which this last inequality holds strictly.

Assume the system $(\boldsymbol{\Phi}_{d}+\boldsymbol{\Phi}_{u})  \ba =\text{vec}(\mathbf{H})$ to be inconsistent, so that
at the optimal point given by
(\ref{eq:opt_point}), we have
$$\text{vec}(\mathbf{H})-(\boldsymbol{\Phi}_{d}+\boldsymbol{\Phi}_{u})(\boldsymbol{\Phi}_{d}+\boldsymbol{\Phi}_{u})^{\dag}\text{vec}(\mathbf{H}) \neq \mathbf{0}$$
or using (\ref{eq:M_ort})
\beq \label{eq:P_M_perp1}
 \mathbf{P}_{\mathcal{M}^{\perp}}\text{vec}(\mathbf{H}) \neq \mathbf{0}.
\eeq
 From
(\ref{eq:sub_incl}) we get
\beq \label{eq:P_M_perp}
\mP_{\mathcal{M}^{\perp}}=\mathbf{P}_{\mathcal{N}+\mathcal{Q}-\mathcal{M}}+\mP_{(\mathcal{N}+\mathcal{Q})^{\perp}},
\eeq
then (\ref{eq:P_M_perp1}) reduces to
\beq \label{eq:P_M_perp22}
\mathbf{P}_{\mathcal{M}^{\perp}}\text{vec}(\mathbf{H})=
\mathbf{P}_{\mathcal{N}+\mathcal{Q}-\mathcal{M}}\text{vec}(\mathbf{H})+\mathbf{P}_{(\mathcal{N}+\mathcal{Q})^{\perp}}\text{vec}(\mathbf{H})
\neq \mathbf{0}.
\eeq
Let us now consider the term $\mathbf{P}_{\mathcal{N}+\mathcal{Q}-\mathcal{M}}\text{vec}(\mathbf{H})$.
 Note that
 \beq \label{eq:p_eq}
 \begin{split}
 \mathbf{P}_{\mathcal{N}+\mathcal{Q}-\mathcal{M}}=&\mathbf{P}_{\mathcal{N}+\mathcal{Q}-\mathcal{M}}(\mathbf{P}_{\mathcal{N}+\mathcal{Q}}+\mathbf{P}_{(\mathcal{N}+\mathcal{Q})^{\perp}}) \\ =& \mathbf{P}_{\mathcal{N}+\mathcal{Q}-\mathcal{M}}\mathbf{P}_{\mathcal{N}+\mathcal{Q}}
 \end{split}
 \eeq
 where we used the identity $\mathbf{P}_{\mathcal{N}+\mathcal{Q}-\mathcal{M}}\mathbf{P}_{(\mathcal{N}+\mathcal{Q})^{\perp}}=\mathbf{0}$,  since from (\ref{eq:range_inclus2})  we have $\mathcal{M} \subset \mathcal{N}+\mathcal{Q}$.
 Hence, defining   $\by:=\mathbf{P}_{\mathcal{N}+\mathcal{Q}}\text{vec}(\mathbf{H})$, we can write
\beq \label{eq:P_ult3}
\begin{split}
\mathbf{P}_{\mathcal{N}+\mathcal{Q}-\mathcal{M}}\text{vec}(\mathbf{H})&=\mathbf{P}_{\mathcal{N}+\mathcal{Q}-\mathcal{M}}\mathbf{P}_{\mathcal{N}+\mathcal{Q}}\text{vec}(\mathbf{H})\\
&=\mathbf{P}_{\mathcal{N}+\mathcal{Q}-\mathcal{M}}\, \by\\
&=\mP_{\mathcal{M}^{\perp}}\by
\end{split}
\eeq
where the last equality follows from  (\ref{eq:P_M_perp}) observing that  $\mathbf{P}_{(\mathcal{N}+\mathcal{Q})^{\perp}} \by=\mathbf{0}$.
Then, replacing  (\ref{eq:P_ult3}) in equation (\ref{eq:P_M_perp22}) we get
\beq \label{eq:P_M_perp5}
\mathbf{P}_{\mathcal{M}^{\perp}}\text{vec}(\mathbf{H})=
\mP_{\mathcal{M}^{\perp}}\by+\mP_{(\mathcal{N}+\mathcal{Q})^{\perp}}\text{vec}(\mathbf{H})
\neq \mathbf{0}.
\eeq
Let us now prove that $\by \neq \mathbf{0}$ if $\text{vec}(\mathbf{H})\neq \mathbf{0}$.
Since from (\ref{eq:phi_ort}) we have $\boldsymbol{\Phi}_{d}^{\dag}\boldsymbol{\Phi}_{u}=\mathbf{0}$, $\boldsymbol{\Phi}_{u}^{\dag} \boldsymbol{\Phi}_{d}=\mathbf{0}$, it also holds $\boldsymbol{\Phi}_{u}^{\dag}\text{vec}(\mathbf{H}_{d})=\mathbf{0}$ and  $\boldsymbol{\Phi}_{d}^{\dag}\text{vec}(\mathbf{H}_{u})=\mathbf{0}$.
Therefore, it results
\beq \label{eq:ineq5}
\begin{split}
& \mP_{\mathcal{N}+\mathcal{Q}}(\text{vec}(\mathbf{H}_{u})+ \text{vec}(\mathbf{H}_{d}))=\\
& (\mP_{\mathcal{N}}+\mP_{\mathcal{Q}})(\text{vec}(\mathbf{H}_{u})+ \text{vec}(\mathbf{H}_{d}))=\\
& \boldsymbol{\Phi}_{d} \boldsymbol{\Phi}_{d}^{\dag}\text{vec}(\mathbf{H}_{d}) +  \boldsymbol{\Phi}_{u} \boldsymbol{\Phi}_{u}^{\dag}\text{vec}(\mathbf{H}_{u})=\bd_d+\bd_{u}
\end{split}
\eeq
where we defined $\bd_d=\boldsymbol{\Phi}_{d} \boldsymbol{\Phi}_{d}^{\dag}\text{vec}(\mathbf{H}_{d})$
and $\bd_u=\boldsymbol{\Phi}_{u} \boldsymbol{\Phi}_{u}^{\dag}\text{vec}(\mathbf{H}_{u})$.
Since from (\ref{eq:orth_cond}) it holds $\bd_u^T \bd_d=0$, this implies that if (\ref{eq:ineq5}) is equal to zero, i.e.
$\bd_d=-\bd_{u}$, then we get  $\bd_u^T \bd_u=0$ and this leads to an absurdum. This implies that we have $\by:=\mP_{\mathcal{N}+\mathcal{Q}}\text{vec}(\mathbf{H})\neq \mathbf{0}$.
Then, given the non-zero vector $\by$, we distinguish the two cases: i)  $\mP_{\mathcal{M}^{\perp}} \by=\mathbf{0}$; ii)  $\mP_{\mathcal{M}^{\perp}} \by \neq \mathbf{0}$. Note that, given a filter $\mH$ to be implemented, both conditions i) and ii) can be easily checked since we know  the projection matrices $\mP_{\mathcal{N}}$, $\mP_{\mathcal{Q}}$ and  $\mP_{\mathcal{M}^{\perp}}$.\\
Let us first consider the case i).
From (\ref{eq:P_M_perp5}) it can be easily derived that  (\ref{eq:dis4}) holds with the equality i.e.
\beq \label{eq:dis}
\parallel \mathbf{P}_{\mathcal{M}^{\perp}} \text{vec}(\mathbf{H}) \parallel^2=
\parallel \mathbf{P}_{(\mathcal{N}+\mathcal{Q})^{\perp}} \text{vec}(\mathbf{H}) \parallel^2.
\eeq
 Using (\ref{eq:M_ort}) and (\ref{eq:sum1})
we get
\beq
\begin{split}
& \parallel \mP_{\mathcal{M}^{\perp}} (\text{vec}(\mathbf{H}_{u})+ \text{vec}(\mathbf{H}_{d}))\parallel^2_F  =\\
& \parallel  (\mI-\mP_{\mathcal{N}}-\mP_{\mathcal{Q}})  (\text{vec}(\mathbf{H}_{d}) + \text{vec}(\mathbf{H}_{u}) )\parallel^2_F\\
& =\parallel  (\mI-\mP_{\mathcal{N}})  \text{vec}(\mathbf{H}_{d}) +
(\mI-\mP_{\mathcal{Q}})\text{vec}(\mathbf{H}_{u}) \parallel^2_F
\end{split}
\eeq
or
\beq
\begin{split}
& \parallel \mP_{\mathcal{M}^{\perp}} (\text{vec}(\mathbf{H}_{u})+ \text{vec}(\mathbf{H}_{d}))\parallel^2_F =\\
&\parallel  (\mI-\mP_{\mathcal{N}})  \text{vec}(\mathbf{H}_{d})\parallel^2_F +\parallel
(\mI-\mP_{\mathcal{Q}})\text{vec}(\mathbf{H}_{u}) \parallel^2_F.
\end{split}
\eeq
Replacing in this last inequality $\mathbf{P}_{\mathcal{Q}}=\boldsymbol{\Phi}_{u} \boldsymbol{\Phi}_{u}^{\dag}$, $\mathbf{P}_{\mathcal{N}}=\boldsymbol{\Phi}_{d} \boldsymbol{\Phi}_{d}^{\dag}$ and (\ref{eq:M_ort}),
we can state that (\ref{eq:ineq_norm}) holds with the equality and this concludes the proof of point i) in the Proposition.
Let us now consider the case ii) where
\beq  \label{eq:P_ineq}
\mP_{\mathcal{M}^{\perp}} \by\neq \mathbf{0}.
\eeq
Then, from (\ref{eq:P_M_perp5}), one easily gets
\beq \label{eq:P_ult5}
\begin{split}
\parallel \mP_{\mathcal{M}^{\perp}} \text{vec}(\mathbf{H}) \parallel^2 &= \parallel \mP_{\mathcal{M}^{\perp}} \by \parallel^2\\& +\parallel\mP_{(\mathcal{N}+\mathcal{Q})^{\perp}}\text{vec}(\mathbf{H})\parallel^2
\end{split}
\eeq
where the first term is non zero from (\ref{eq:P_ineq}).
Therefore,  it results
\beq \label{eq:P_ult6}
\begin{split}
\parallel \mP_{\mathcal{M}^{\perp}} \text{vec}(\mathbf{H}) \parallel^2 & > \parallel \mP_{(\mathcal{N}+\mathcal{Q})^{\perp}}\text{vec}(\mathbf{H})\parallel^2
\end{split}
\eeq
and this concludes the  proof of point ii) in the Proposition.
It remains to prove point (b).
Let us assume that $\mP_{\mathcal{M}} \text{vec}(\mathbf{H})=\text{vec}(\mathbf{H})$ so that from (\ref{eq:l2}) there exists at least a solution $\ba^{\star}$ such that $f(\ba^{\star})=0$, i.e. it holds
\beq
\left[\begin{array}{ll} \bh^s\\
\bh^I
\end{array}\right]=\left[\begin{array}{ll} \boldsymbol{\Phi}_s \ba^{\star}\\
 \boldsymbol{\Phi}_I \ba^{\star}
\end{array}\right].
\eeq
Then, we get $f_s(\ba^{\star})=f_I(\ba^{\star})=f(\ba^{\star})=0$. This concludes the proof of point (b).

\bibliographystyle{IEEEbib}
\bibliography{reference}

\end{document}